\documentclass[onecolumn, 12pt, draftclsnofoot]{IEEEtran}
\usepackage{amsmath,amssymb, mathtools,amsthm,gensymb }
\usepackage{relsize}

\usepackage{booktabs}
\usepackage{enumerate}
\usepackage{graphicx}
\usepackage{epsfig}
\usepackage{float}
\usepackage[export]{adjustbox} 

\DeclareMathOperator*{\argmax}{\arg\!\max}
\usepackage{blindtext}
\usepackage{algorithm,algcompatible}
\usepackage{varwidth}
\algnewcommand\INPUT{\item[\textbf{Input:}]}%
\algnewcommand\OUTPUT{\item[\textbf{Output:}]}%
\theoremstyle{definition}

\theoremstyle{definition}
\newtheorem{theorem}{Theorem}
\theoremstyle{remark}
\newtheorem*{remark}{\textbf{Remark}}
\newtheorem{lemma}{\textbf{Lemma}}
\newtheorem{proposition}{Proposition}
\usepackage{layouts}
\usepackage{color}
\usepackage{xcolor}
\usepackage{cite}
\usepackage{enumitem}
\usepackage{printlen}

\usepackage{setspace}
\usepackage{xcolor}
\usepackage{float}
\usepackage{caption}
\usepackage{subcaption}
\usepackage{multirow}
\allowdisplaybreaks
\algdef{SE}[SUBALG]{Indent}{EndIndent}{}{\algorithmicend\ }%
\algtext*{Indent}
\algtext*{EndIndent}

\usepackage{tikz}
\usetikzlibrary{positioning}
\usetikzlibrary{shapes.geometric, arrows}
\usetikzlibrary{chains}

\tikzstyle{startstop} = [rectangle, rounded corners, minimum width=1.1cm,  minimum height=1.1cm,text centered, draw=black, fill=white!30, rounded corners=15pt, text width=4cm]
\tikzstyle{process} = [rectangle, minimum width=1cm, minimum height=1cm, text centered, draw=black, fill=white!30, rounded corners, text width=3cm]
\tikzstyle{Test} = [diamond, minimum width=0.00001cm, minimum height=0.00001cm, text centered, draw=black, fill=white!30, text width=2cm]
\tikzstyle{arrow} = [thick,->,>=stealth]

\newcommand{\norm}[1]{\left\lVert #1 \right\rVert}
\newcommand{\revised}[1]{\textcolor{black}{#1}}
\begin{document}
	\title{Resource Allocation in NOMA-based Self-Organizing Networks using Stochastic Multi-Armed Bandits}
	\author{Marie-Josepha~Youssef,~\IEEEmembership{Student~Member,~IEEE,}
	Venugopal V. Veeravalli,~\IEEEmembership{Fellow,~IEEE,}  Joumana~Farah,~\IEEEmembership{Member,~IEEE,}
	Charbel~Abdel~Nour,~\IEEEmembership{Senior Member,~IEEE,} and ~Catherine~Douillard,~\IEEEmembership{Senior~Member,~IEEE}\\\vspace*{-1cm}
	\thanks{M. J. Youssef, C. Abdel Nour and C. Douillard are with IMT Atlantique, LabSTICC, UBL, F-29238 Brest, France, (e-mail: marie-josepha.youssef@imt-atlantique.fr; charbel.abdelnour@imt-atlantique.fr; catherine.douillard@imt-atlantique.fr).} 
	\thanks{ V. V. Veeravalli is with the ECE Department, University of Illinois at Urbana–Champaign, Urbana, IL 61801 USA, and also with the Coordinated Science Laboratory, University of Illinois at Urbana–Champaign, Urbana, IL 61801 USA (e-mail: vvv@illinois.edu).}
	\thanks{ J. Farah is with the Department of Electricity and Electronics,	Faculty of Engineering, Lebanese University, Roumieh, Lebanon (e-mail: joumanafarah@ul.edu.lb).}
	\thanks{This work has been funded with support from the UBL, the GdR ISIS, the Lebanese University, and the US National Science Foundation SpecEES program under grant number 1730882, throughout the University of Illinois at Urbana-Champaign (UIUC). }}

	\maketitle
	\vspace{-0.7cm}
		\begin{abstract}
				\vspace{-0.3cm}
		\revised{To achieve high data rates and better connectivity in future communication networks, the deployment of different types of access points (APs) is underway. In order to limit human intervention and reduce costs, the APs are expected to be equipped with self-organizing capabilities. Moreover, due to the spectrum crunch, frequency reuse among the deployed APs is inevitable, exacerbating the problem of inter-cell interference (ICI).  Therefore, ICI mitigation in self-organizing networks (SONs) is commonly identified as a key radio resource management mechanism to enhance performance in future communication networks.  With the aim of reducing ICI in a SON, this paper proposes a novel solution for the uncoordinated channel and power allocation problems. Based on the multi-player multi-armed bandit (MAB) framework, the proposed technique does not require any communication or coordination between the APs.} The case of varying channel rewards across APs is considered. In contrast to previous work on channel allocation using the MAB framework, APs are permitted to choose multiple channels for transmission. Moreover, non-orthogonal multiple access (NOMA) is used to allow multiple APs  to access each channel simultaneously. This results in an MAB model with varying channel rewards, multiple plays and non-zero reward on collision. The proposed algorithm has an expected regret in the order of  $\mathcal{O}(\log^2T)$, which is validated by  simulation results. \revised{Extensive numerical results also reveal that the proposed technique significantly outperforms the well-known upper confidence bound (UCB) algorithm, by achieving more than a twofold increase in the energy efficiency.}
	\end{abstract}
	\vspace{-0.5cm}
	\begin{IEEEkeywords}
		\vspace{-0.3cm}
		Uncoordinated channel and power allocation, MAB with multiple plays and non-zero reward on collision, varying reward distribution, NOMA, self-organizing networks.
	\end{IEEEkeywords}

\section{Introduction}
Future cellular communication networks are expected to support a myriad of new applications and services conceived for both traditional human-type devices and for the growing number of machine-type devices \cite{7397856}. To meet the exponential growth in connectivity and mobile traffic, new technologies are needed. Among these new technologies, the deployment of different types of access points (AP), e.g., small base-stations (SBS), pico-cells, femto-cells, relays, etc., is of particular importance, since APs can offload mobile traffic from highly congested macro-base stations (MBS) \cite{6525592}. To limit human intervention and reduce planning and maintenance costs, APs can be equipped with self-organizing capabilities \cite{sonReference}, allowing them to optimize their resource use in a distributed manner. APs normally have a lower transmit power budget and a smaller coverage range when compared to traditional MBSs. However, thanks to their denser deployment,  APs benefit from the ability to consume less transmit power, leading to significant gains in power consumption as was shown in \cite{6215539,JA_VehTech}.  That said, by introducing APs into the network, the problem of inter-cell-interference (ICI) is aggravated, necessitating the application of adequate resource allocation algorithms to limit the interference \cite{5441362}.

The problem of ICI in self-organizing networks (SON) was extensively studied in the literature. In \cite{6945909}, the weighted sum-rate of the system is optimized through ICI coordination between SBSs. The authors adopt a blanking method where at the level of each SBS, some wireless channels are not used to mitigate the ICI. In \cite{7070660}, an algorithm for ICI coordination between SBSs based on asynchronous inter-cell signaling is proposed. The authors of \cite{7343514} propose an algorithm based on a semi-static frequency allocation to mitigate ICI and enhance the performance of cell-edge users. The proposed solutions of \cite{6945909, 7070660, 7343514 } rely on explicit communication between the distributed SBSs to mitigate ICI, resulting in excessive signaling among SBSs. To limit signaling overhead, decentralized algorithms, based on reinforcement learning, are preferred. 

\revised{The use of reinforcement learning in wireless communications has recently garnered significant attention \cite{7476897}.} The related framework of multi-player multi-armed bandits (MAB)  \cite{lattimore} has also been widely used to study  multiple problems in wireless communication systems ranging from SON \cite{6140047,6240019,7248837}, to uncoordinated spectrum access \cite{5738217,5535151,magesh, meghana}, to fast uplink grant allocation \cite{9075198}, to unmanned-aerial vehicles positioning and path-planning \cite{8669870}. In the context of SON, in \cite{6140047,6240019}, a solution is proposed based on the stochastic MAB framework to allow SBSs to partition efficiently the available frequency resources in an effort to mitigate ICI. In \cite{7000863}, a method based on learning authomata is proposed where femto-cells adjust their resource use based on the feedback received from users. In \cite{7248837}, the authors resort to the EXP3 algorithm from the adversarial MAB framework to mitigate the  ICI while allowing  each base-station (BS) to access multiple frequency bands. \revised{The work in \cite{8336932} proposes a data-driven approach based on the MAB framework to address the ICI problem in heterogeneous networks (HetNets).} The MAB framework was also widely used to study the opportunistic and the uncoordinated spectrum access problems.
 For example, in \cite{5738217}, \cite{5535151} and \cite{8676344},  the MAB model is used to study the opportunistic spectrum access problem in cognitive radio networks where secondary users compete to access the part of the spectrum not occupied by primary users. In addition to studying  the opportunistic channel access problem,  in \cite{8676344}, the authors also solve the distributed power allocation problem. 
 In contrast to opportunistic channel access,  the authors of \cite{magesh}, \cite{meghana} and \cite{8902878} employ MAB to study the uncoordinated spectrum access problem without distinguishing between the users. The distributed power control problem is studied in \cite{8902878},  and solutions are proposed based on the upper-confidence-bound (UCB) algorithm, and on the $\epsilon$-greedy algorithm. In \cite{6953296}, the channel and power allocation problem in a device-to-device  system is modeled using the MAB framework. A game-theoretic solution based on the potential game framework is proposed to minimize regret of users.

With the exception of \cite{7248837}, all previous work on wireless communications solutions based on MABs assumes that each player chooses one channel at each timeslot. However, removing this assumption is expected to improve performance for the players if a suitable algorithm is formulated, especially for the case of a SON. Indeed when an AP can access multiple channels simultaneously, an increase in both, the probability of a successful transmission and the achieved reward or rate is observed, allowing the AP to serve more end-users. Moreover, with the exception of \cite{magesh, 6953296, meghana}, all previous work based on MABs considered a zero reward for multiple players accessing the same channel.  By alleviating this assumption and adopting non-orthogonal multiple access (NOMA), system performance is expected to further improve.

From an information-theoretical point of view, it is well-known that non-orthogonal user multiplexing using superposition coding at the transmitter and proper decoding techniques at the receiver not only outperforms orthogonal multiplexing, but is also optimal in the sense of achieving the capacity region of the downlink broadcast channel \cite{tse}. As a result, NOMA emerged as a promising multiple access technology for 5G systems \cite{6666209,8869799,tvt_mjy}. NOMA allows multiple users to be scheduled on the same time-frequency resource by multiplexing them in the power domain. At the receiver side, successive interference cancellation (SIC) is performed  to retrieve superimposed signals.

{To limit the ICI in a SON, studying the resource allocation in the fronthaul portion of the network is of utmost importance \cite{6140047,6240019,7248837}. When coupled with optimizing the resource allocation in the backhaul link, optimizing the fronthaul portion leads to significant performance gains \cite{8891385, tvt_mjy}.}

  In this paper, we consider the fronthaul part of a self-organizing wireless network where multiple APs aim at organizing their uplink transmissions with a central unit in a distributed manner. Both the uncoordinated channel access and the distributed power control problems are studied. A solution based on the MAB framework, which does not necessitate any coordination or communication between APs, is proposed. The considered setting is closest to the ones studied in \cite{magesh} and \cite{got}, where a game-theoretic approach is used to solve the uncoordinated channel access problem. Our study extends that of \cite{magesh} and \cite{got} by allowing each AP to access multiple channels simultaneously and by proposing a model for the distributed power control problem. The main contributions of this paper can be summarized as follows:
\begin{itemize}
	\item A two-phase algorithm based on the MAB framework, extending the work in \cite{got, magesh}, is proposed for the uncoordinated channel access and distributed power control problems.
	\item For the first phase, i.e., the uncoordinated channel access phase, in addition to considering  varying channel rewards between APs, each AP is allowed to simultaneously access multiple channels. This is in contrast to the work in \cite{magesh}  and \cite{got}  where each player accesses one channel in a timeslot. Moreover, each channel can accommodate multiple APs at once using NOMA, leading to a multi-player MAB  problem with varying player rewards, multiple plays and non-zero reward on collision. 
	\item For the power control phase, varying power level rewards between APs are considered and an algorithm to solve the power control problem on each channel is proposed.
	\item The proposed technique is shown to achieve a sublinear regret of $\mathcal{O}(\log^2 T)$. In addition, simulation results validating the theoretical results and the performance of the proposed technique are presented.
	\item   To the best of our knowledge, this is the first work that studies the uncoordinated channel access and the distributed power control problems in a SON network, using both NOMA and the multi-player MAB framework with varying channel rewards across users, multiple plays, and non-zero reward on collision.
\end{itemize}


The rest of this paper is organized as follows. The system model is presented in section \ref{sec:systemModel}. In sections \ref{sec:solution}, \ref{sec:regret}, \ref{sec:expPhase} and \ref{sec:matPhase}, the proposed algorithm is presented along with an analysis of the system-wide regret. Simulation results are presented in section \ref{sec:results} and conclusions in section \ref{sec:conc}.

\section{System Model}\label{sec:systemModel}

Consider the  uplink of a cellular system as shown in Fig.\ \ref{fig:sysModel} where $K$ APs aim to organize their communications with an MBS serving as gateway to the core network, over  $M$ available wireless channels, in an uncoordinated manner. The communication occurs over a finite time horizon $T$ that may not be known in advance to the APs. At each timeslot $t$,  every AP $k$  chooses $N$ channels, adjusts its transmission power, and transmits over the chosen channels. \revised{Note that the proposed solution can be easily extended to the case where each AP $k \in \mathcal{K}$ chooses $N_k$ channels at each timeslot, where $1\leq N_k \leq M$.}
We assume that NOMA is employed, enabling multiple APs to choose the same channel for communication and  achieve a non-zero rate. 
\begin{figure}[!h]
	\begin{center}
		\vspace{-0.3cm}
		\includegraphics[height=5.5cm,keepaspectratio]{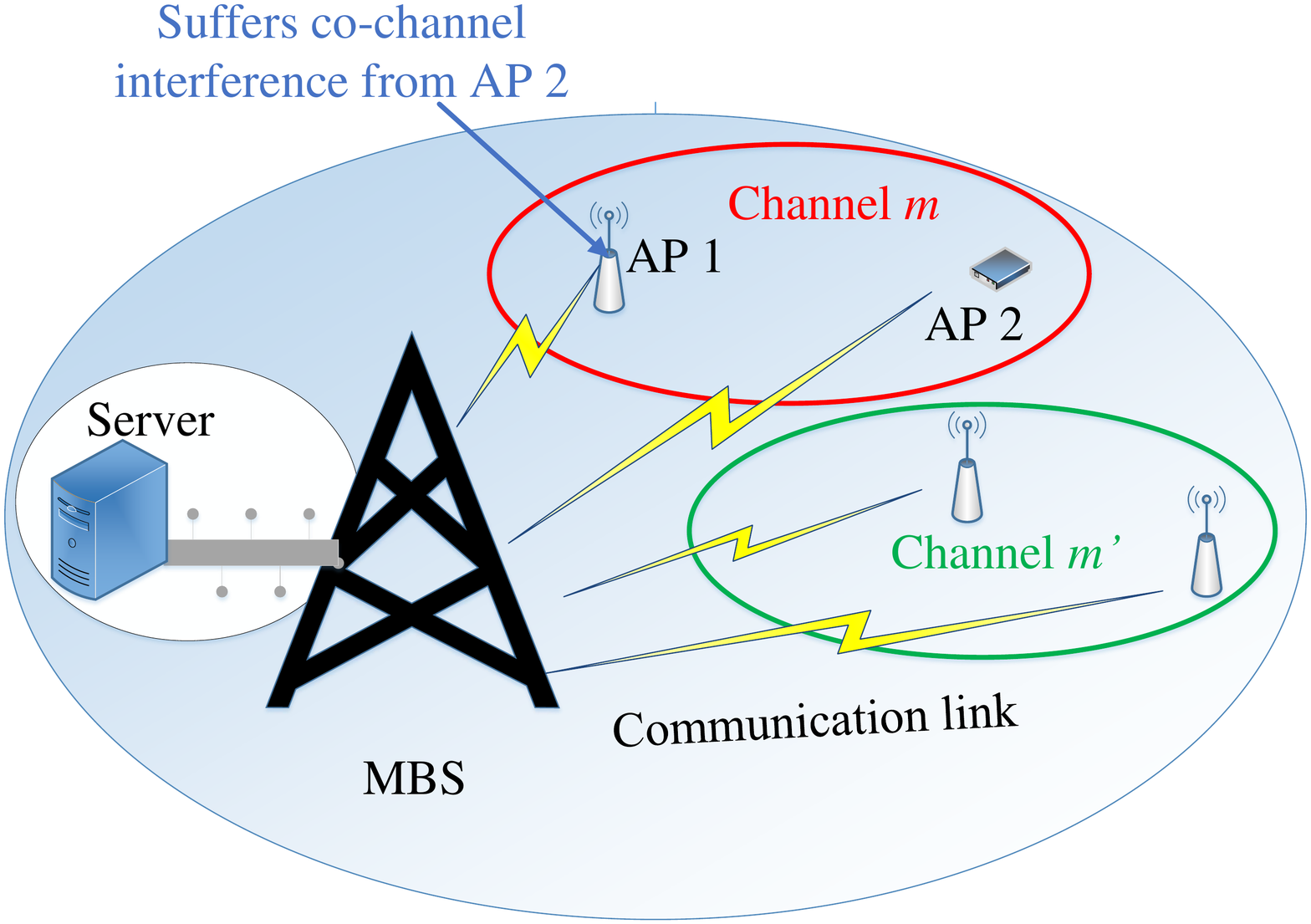}  
		\vspace{-0.2cm}
		\caption{\label{fig:sysModel} {System Model.}}
	\end{center}\vspace{-0.8cm}
\end{figure}
That said, if two or more APs choose the same channel, the received power levels of these APs must be  different at the receiving BS level in the core network, to enable SIC decoding at the receiver side. To ensure the reception of different received power levels for the signals transmitted by the APs, we generalize the uplink NOMA power allocation model introduced in \cite{8085106}, where for a constant SINR requirement, $L$ received power levels, ensuring the SINR requirement for $L$ users scheduled on the same channel, are calculated. In this work, we extend the study of \cite{8085106} to allow for $L$ distinct SINR requirements per channel, $\boldsymbol{\Gamma}=\{\Gamma_1,\ldots,\Gamma_L\}$, sorted by decreasing order. { Note that allowing for distinct SINR levels inherently encompasses the special case of constant SINR levels.}  An AP $k$ choosing SINR requirement $\Gamma_l$ over channel $m$ achieves the following uplink data rate: 
\begin{equation}
R_{k,m,l}=\log_2\left(1+\Gamma_l\right),  
\end{equation}
where $\Gamma_l$ is given by:
\begin{equation}\label{eq:gammaLevel}
\Gamma_l=\frac{v_l}{V_l+N_0B_c}.
\end{equation}
In Eq.\ (\ref{eq:gammaLevel}), $v_l$ is the received power level of AP $k$, the expression of which is given in Section \ref{subsec:powerModel}, $N_0$ is the noise power spectral density and $B_c$ the channel bandwidth.  At the receiver side, when the AP transmissions are received with different power levels, SIC is employed to decode the received messages in a descending order. In other words, the AP choosing the highest SINR requirement $\Gamma_1$, and consequently the highest received power level $v_1$, suffers  interference from all APs choosing a lower SINR requirement. Once decoded, the signal of  the AP choosing $\Gamma_1$ is removed using SIC before decoding the remaining messages.  Hence, variable $V_l$ of Eq.\ (\ref{eq:gammaLevel}) is the power level of the interfering transmissions, not canceled with SIC, expressed as: $V_l=\sum_{l'=l+1}^L v_{l'}$.
To limit the decoding complexity at the receiving BS in the core network, as well as the error propagation in SIC, the number of APs allowed to access a channel and achieve a non-zero rate is limited to $\beta$, such that $\beta M\geq K N$. \revised{Note that in the case of a varying number of chosen channels across users, this last condition becomes $\beta M\geq \sum_{k\in\mathcal{K}} N_k$.}  It is assumed that when an AP $k$ accesses a channel $m$, $k$ knows the total number of APs currently accessing channel $m$. No \emph{a priori} knowledge of  the channel gain experienced over each channel is assumed. Moreover, these channel gains are distinct for each AP. To solve the channel and power allocation problems in an uncoordinated manner, we proceed in two steps, the first, of length $T_C$, dedicated to  channel allocation and the second, of length $T_P$, dedicated to  power allocation. Note that both $T_C$ and $T_P$ may not be known to the APs.

\subsection{Uncoordinated channel allocation} \label{subsec:channelModel}
To allow each AP to access $N$ channels simultaneously in a NOMA manner, the problem of uncoordinated multiple access is modeled as a stochastic multi-player MAB problem with multiple plays and non-zero reward on collision. The set of players is the set of APs $\mathcal{K}$ and the set of arms is the set of channels $\mathcal{M}$. The action of each AP $k$ at each timeslot $t$ is $\boldsymbol{a}_k^t\in\{0,1\}^{M\times 1}$ such that $a_k^t (m)=1  $ if AP $k$ pulls channel $m$ at timeslot $t$. Moreover, $\sum_{m=1}^M a_k^t (m)=N,~\forall k\in\mathcal{K}.$ The action space of each AP $k$, $\mathcal{A}_k$, consists of all possible combinations of $N$ channels, hence $|\mathcal{A}_k|=\binom{M}{N}$. Let $\boldsymbol{a}^t=\{\boldsymbol{a}_1^t,\ldots,\boldsymbol{a}_K^t\}$ denote the strategy profile of  all APs in timeslot $t$. 
\revised{Upon choosing an action $\boldsymbol{a}_k^t \in \boldsymbol{a}^t$, AP $k$ receives the following  average reward:} 
\begin{equation}
g_k^t(\boldsymbol{a^t})=\sum_{m=1}^M a_k^t (m) \mu_M(k,m, k_m),
\end{equation}
where $k_m$ is the number of APs choosing channel $m$ at timeslot $t$. Variable $\mu_M(k,m, k_m)$ is the mean reward of AP $k$ over channel $m$  when $k_m$ APs access it. \revised{Note that the actual value of the received reward by AP $k$ when choosing channel $m$ at timeslot $t$ is drawn from a uniform distribution with mean $\mu_M(k,m, k_m)$.}

 We assume that the mean reward of AP $k$ when accessing channel $m$ alone is equal to the normalized \revised{average} channel gain of AP $k$  over channel $m$, i.e.,
\begin{equation}\label{eq:rewardNorm}
\mu_M(k,m,1)=h_{k,m}/\mu_M^{max},
\end{equation} 
where $h_{k,m}$ is the \revised{average} channel gain of AP $k$ over channel $m$ and $\mu_M^{max}=\max\limits_{k\in\mathcal{K}, m \in \mathcal{M}}h_{k,m}$. {Note that it is assumed that the  BS at the core network performs channel estimation on the received signals from all APs. Hence,  the average channel gains $h_{k,m}, \forall k \in \mathcal{K}, \forall m \in \mathcal{M}$ are assumed to be perfectly known by the receiving BS.}
For $1<k_m\leq \beta$, the mean reward of an AP must account for the added interference brought by the $(k_m-1)$ other APs scheduled on the same channel $m$. Ideally, the mean reward should take into account the interference brought by each particular AP. However, that would result in a prohibitive complexity since any channel, for each $1<k_m\leq \beta$, would have $\binom{K-1}{k_m}$ distinct reward values. To simplify the analysis, in this work, we assume that the mean reward for $1<k_m\leq \beta$, is a decreasing function of the number of interfering APs on the same channel. In other words,
\begin{equation}
\mu_M(k,m,k_m)={\mu_M(k,m,1)}/{k_m}.
\end{equation}
When $k_m>\beta$, $\mu_M(k,m, k_m)=0$. The normalization in Eq.\ (\ref{eq:rewardNorm}) leads to: $\mu_M(k,m, k_m) \in [0,1]$  for every AP $k \in \mathcal{K}$, on every channel $m\in\mathcal{M}$ and for every number of APs $k_m \in [\beta]$. Hence, $g_k^t(\boldsymbol{a^t}) \in [0,N]$.
 
\revised{In addition to receiving the achieved rewards, we assume that the feedback received by each AP $k$ from the MBS includes the total number of APs simultaneously accessing its chosen channels. In other words, for all channels $m$ such that $a_k^t(m)=1$, AP $k$ receives the total number of APs accessing channel $m$, i.e., receives $k_m=\sum_{k \in \mathcal{K}}a_k^t(m)$. {Note that this assumption is necessary for the correct estimation of the mean rewards, allowing APs to learn and settle on the optimal allocation. Moreover, since $\beta$ is normally kept small, feeding back to each AP $k$ the total number of APs simultaneously accessing its chosen channels  requires only a few bits.}}
 
APs make their decisions in a distributed manner observing neither the channels chosen by other APs nor the rewards received by other APs. Each AP $k$ can only observe the reward it gets on each of its chosen channels. 
Our aim is to propose a distributed algorithm allowing APs to organize their transmissions on the available channels, without communicating together, in such a way as to maximize the sum reward of the system. By definition,  the action profile  yielding the highest sum reward $\boldsymbol{a^*}$ is given by:
\begin{equation}
\boldsymbol{a^*}=\argmax_{\boldsymbol{a}\in \mathcal{A}} \sum_{k=1}^K\sum_{m=1}^M  a_k (m) \mu_M(k,m, k_m),
\end{equation}
\revised{where $\mathcal{A}$ is the action space of all APs, i.e., $\mathcal{A}=\prod_{k\in \mathcal{K}}\mathcal{A}_k$.}

The expected regret incurred during $T_C$ is the difference between the achieved reward  when playing $\boldsymbol{a^*}$ at all timeslots, and the actually achieved reward  by the learning players during the $T_C$ timeslots \cite{lattimore}. In our case, it is given by: 
\begin{equation}
\revised{\bar R= T_C\sum_{k,m} a^*_k (m) \mu_M(k,m, k^*_m) -\mathbb{E}\left(\sum\limits_{t=1}^{T_C}\sum_{k,m}  a_k (m) \mu_M(k,m, k_m)\right),}
\end{equation}
where $k^*_m$ is the optimal number of APs scheduled over channel $m$ under $\boldsymbol{a^*}$.

After $T_C$ timeslots, the APs receive a signal from the core network to terminate the channel allocation phase. At the end of the channel allocation phase, at most $\beta$ APs are scheduled over each channel $m\in \mathcal{M}$.  Moreover,  as an outcome of this first phase, each AP $k$ computes an estimate of its average channel gain over each channel $m$, denoted by $\hat{h}_{k,m}$. 

\subsection{Distributed Power Allocation}\label{subsec:powerModel}
Once settled over their chosen channels, the APs receive a signal from the core network to move to the power allocation stage. Since different frequency bands are allocated to  different channels, power allocation over each channel $m$ can be done independently of other channels $m' \in \mathcal{M}\setminus{\{m\}}$. In the following, we will focus on the power allocation over channel $m\in \mathcal{M}$, where the set of scheduled APs is $\mathcal{K}_m$. 

To simplify the distributed power allocation, we assume that each AP chooses, for each of its allocated channels, one SINR level among a fixed set of  $L\geq \beta$ available SINR levels, with $\boldsymbol \Gamma$ being the set of pre-determined available SINR levels. The AP then calculates the necessary power level $v_l$ for the chosen SINR level $\Gamma_l$. For successful SIC decoding, each power level can support one AP only. In other words, if multiple APs choose the same power level, SIC fails and the signals of all $K_m$ APs are not decodable.
Inspired by \cite{8085106}, it can be shown that, to satisfy Eq.\ (\ref{eq:gammaLevel}), the power level $v_l$ must be set as:
\begin{equation}\label{eq:powerLevel}
v_l=\Gamma_l N_0 B_c \prod\limits_{l'=l+1}^L\left(\Gamma_{l'}+1\right).
\end{equation} 
\revised{Note that the expression of $v_l$ is obtained by proceeding backwards and by induction from $v_L=\Gamma_L N_0B_c$.}

{The expression of $v_l$ ensures the SINR requirement $\Gamma_l$ when considering that an AP chooses each subsequent SINR requirement, hence the worst case scenario. Note that our setting allows for similar SINR levels. However, for similar or distinct SINR levels, the power levels chosen by APs need to be distinct to allow for SIC decoding.}

To ensure SIC stability, i.e., successful decoding of the received signals in descending order \cite{7982784}, the distributed power control scheme must ensure that the power of each signal scheduled for decoding at the BS is larger than the received power of the interference generated by the combination of the  remaining signals, i.e., $v_l> V_l$. From Eq.\ (\ref{eq:powerLevel}), the power level $v_l$ depends on the associated SINR level $\Gamma_l$ as well as on the interfering SINR levels $\Gamma_{l'}, l'=l+1,\ldots,L$. 
\begin{proposition}To ensure SIC stability, the available SINR levels must satisfy:
\begin{equation}\label{eq:sicStability}
\Gamma_l>\frac{2^{(L-l-1)}\times\Gamma_L}{\prod\limits_{l'=l+1}^L\left(\Gamma_{l'}+1\right)}.
\end{equation}
\end{proposition}
\begin{proof}
By proceeding backwards, to get $v_{L-1}>v_L$, the following must hold:
\begin{equation}\label{eq:sicStabL-1}
\Gamma_{L-1}>\frac{\Gamma_L}{\Gamma_L+1}=\frac{2^{(L-(L-1)-1)}\Gamma_L}{\Gamma_L+1}.
\end{equation}	
Similarly, to get $v_{L-2}>v_{L-1}+v_L$, the following must hold:
\revised{
\begin{equation}
\Gamma_{L-2}>\frac{\Gamma_{L-1}(\Gamma_L+1)+\Gamma_L}{(\Gamma_{L-1}+1)(\Gamma_L+1)}\overset{\text{(a)}}>\frac{\frac{\Gamma_L}{\Gamma_L+1} (\Gamma_L+1)+\Gamma_L}{(\Gamma_{L-1}+1)(\Gamma_L+1)}>\frac{2\Gamma_L}{(\Gamma_{L-1}+1)(\Gamma_L+1)}=\frac{2^{(L-(L-2)-1)}\Gamma_L}{\prod\limits_{l'=L-1}^L(\Gamma_{l'}+1)},
\end{equation}
where (a) follows from Eq.\ (\ref{eq:sicStabL-1}).}

To get $v_l>V_l=\sum\limits_{l'=l+1}^L v_{l'}$, assume that Eq.\ (\ref{eq:sicStability}) holds.
By induction, to get $v_{l-1}>\sum\limits_{l'=l}^L v_{l'}$, we must have:
\begin{equation}
\Gamma_{l-1}>\frac{2^{(L-(l-1)-1)}\Gamma_L}{\prod\limits_{l'=l}^L(\Gamma_{l'}+1)}.
\end{equation}
\end{proof}
Knowing the available SINR levels, each AP $k \in \mathcal{K}_m$ calculates the associated received power levels using Eq.\ (\ref{eq:powerLevel}). Then, using the estimated average channel gain over $m$, $\hat{h}_{k,m}$, AP $k \in \mathcal{K}_m$  calculates the necessary transmit power for each power level $v_l$, $p_{k,m,l}$, according to:
\begin{equation}
p_{k,m,l}={v_l}/{\hat{h}_{k,m}^2}.
\end{equation}
Each AP is assumed to have a power budget per channel $P_k^{m}$. Hence, AP $k$ 
can transmit over channel $m$ using  power level $v_l$ if $p_{k,m,l}\leq P_k^m$. 
AP $k \in \mathcal{K}_m$ builds the set of possible power levels, $\mathcal{P}_{k,m}^a$, where $ \mathcal{P}_{k,m}^a=\{v_l | ~  p_{k,m,l}\leq P_k^m, l\in[L]\}$.
Note that the set of possible power levels are AP-dependent because of their dependency on the estimated average channel gain of each AP, $\hat{h}_{k,m}$, and on the AP power budget. 

The power allocation among APs on the same channel consists of APs  choosing SINR levels, and hence received power levels, in a distributed manner, and without any inter-AP coordination. Since APs choosing  the same SINR level result in an unsuccessful SIC decoding, the APs must aim at organizing their transmissions using different SINR levels. For this purpose, the power allocation on each channel is modeled using the MAB framework with single play and zero-reward on collision. Over channel $m$, the set of players is $\mathcal{K}_m$ and the set of arms is the set of power levels \revised{$\mathcal{VL}=\{v_l, l=1,\ldots,L\}$}. Since $L=|\mathcal{VL}|\geq \beta \geq K_m=|\mathcal{K}_m|$, a solution where each AP accesses one power level, without collision, is achievable. 
At each timeslot, each AP $k \in \mathcal{K}_m$ chooses an action $a_{k,m}^t$, i.e., a power level $v_l \in \mathcal{P}_{k,m}^a$, and transmits using $p_{k,m,l}$. 
The action space of AP $k$ is 
$\mathcal{P}_{k,m}^a$.  Let $\boldsymbol{a}_m^t$ denote the strategy chosen by all APs in $\mathcal{K}_m$ over channel $m$ at timeslot $t$. \revised{ Upon choosing action $a_{k,m}^t \in \boldsymbol{a}_m^t$, AP $k$ receives the following average reward on channel $m$:}
\revised{\begin{equation}
g_{k,m}^t(\boldsymbol{a}^t)=\mu_P(k, m, a_{k,m}^t)\eta(\boldsymbol{a}_m^t), 
\end{equation}}
where $\mu_P(k, m, a_{k,m}^t)$ is the reward of AP $k$ when choosing $a_{k,m}^t$. \revised{Note that the actual value of the received reward by AP $k$ when choosing action $a_{k,m}^t$ on channel $m$ at timeslot $t$ is drawn from a uniform distribution with mean $\mu_P(k, m, a_{k,m}^t)$.}

The mean reward $\mu_P(k, m, a_{k,m}^t)$ is chosen in a way to strike a trade-off between SINR maximization and transmit power minimization. Therefore, it is set as:
\begin{equation}
\mu_P(k,m,  a_{k,m}^t=v_l)= w_k^1 \frac{\Gamma_l}{\Gamma_{max}}+w_k^2\frac{1}{p_{k,m,l}~\max\limits_{k,m,l}(\frac{1}{p_{k,m,l}})},
\end{equation}
where $w_k^1$ and $w_k^2$ are weight parameters relative to AP $k \in \mathcal{K}_m$ satisfying $w_k^1+w_k^2=1$. \revised{ The variable $\Gamma_{max}$ is the highest available SINR, i.e., $\Gamma_{max}=\Gamma_1$.} Note that $\mu_P(k,m, a_{k,m}^t) \in [0,1]$ and is not known by the AP in advance. \revised{Let $\mathcal{N}^m_{v_l}(\boldsymbol{a}_m^t)$ be the set of APs choosing power level $v_l$ at timeslot $t$, i.e., $\mathcal{N}^m_{v_l}(\boldsymbol{a}_m^t)=\{k \in \mathcal{K}_m ~|~ a_{k,m}^t=v_l\}$. The variable $\eta(\boldsymbol{a}_m^t)$ is the collision indicator of the strategy profile of all APs, $\boldsymbol{a}_m^t$, i.e., 
$\eta(\boldsymbol{a}_m^t)=
1$ if $|\mathcal{N}^m_{a_{k,m}^t=v_l}(\boldsymbol{a}_m^t)|\leq 1,  \forall ~ v_l \in \mathcal{VL}$, and 0 otherwise.} \revised{Note that no feedback regarding the  value of the collision indicator $\eta(\boldsymbol{a}_m^t)$ is necessary. In fact, in the case of collisions, the MBS  does not have to return any feedback to the colliding APs who will assume a zero reward is achieved. When no collision takes place, the MBS returns only the value  of the mean reward to the AP since the collision indicator is equal to one in the case of no collision.}
\vspace{0.05cm}

APs choose power levels in a distributed manner without any coordination, 
with each AP only  observing the reward received on the chosen power level. 
The proposed power allocation scheme aims at maximizing the sum reward of the system. Let $\boldsymbol{a}_{m}^{*P}$ be the action profile yielding the highest sum reward over channel $m$:
\revised{\begin{equation}
\boldsymbol{a}_{m}^{*P}=\argmax\limits_{\boldsymbol{a}_m\in \mathcal{P}^a_m}\sum\limits_{k\in \mathcal{K}}\mu_P(k,m, a_{k,m}^t)\,\eta(\boldsymbol{a}_m^t),
\end{equation}}
\revised{where $\mathcal{P}^a_m$ is the action space of all APs scheduled on channel $m$, i.e., $\mathcal{P}^a_m=\prod_{k\in \mathcal{K}_m}\mathcal{P}^a_{k,m}$.}

The expected regret incurred during the time horizon $T_P$ over all $M$ channels is given by:
\begin{equation}
\revised{\bar{R}_p=\sum\limits_{m\in\mathcal{M}}\left \{T_P\sum\limits_{k\in \mathcal{K}}\mu_P(k,m, a^{*P}_{k,m})\right.-\left.\mathbb{E}\left(\sum\limits_{t=1}^{T_P}\sum\limits_{k}\mu_P(k,m, a_{k,m}^t)\,\eta(\boldsymbol{a}_m^t)\right)\right\}.}
\end{equation}

\section{Proposed Solution}\label{sec:solution}
\subsection{Proposed Algorithm for the Channel Allocation Problem}
Since the time horizon $T_C$ is not necessarily known in advance, the proposed solution, presented in Algorithm  \ref{alg1}, proceeds in epochs, each epoch consisting of  three phases, namely, \textit{exploration, matching and exploitation}. The exploration phase aims at estimating the previously unknown means of each channel, as well as the number of APs competing for system resources. During this phase, each AP uniformly accesses one channel at a time to estimate its mean reward. AP $k$ accessing channel $m$ gets as feedback the achieved reward on $m$ as well as the total number of APs simultaneously accessing channel $m$. This phase runs for a constant number of timeslots given by $T_C^0$. Upon termination, all APs have an estimate $\boldsymbol{\hat{\mu}}_M$ of the means of the channels and of the channel gain experienced over each channel. Each AP also calculates an estimate of the number of APs $\hat{K}$, as was done in \cite{meghana}. These estimated means and number of APs are used in the second phase of the algorithm where APs  play a non-cooperative game with the aim of maximizing the achieved sum rewards. The  estimated reward means are taken to be the actual utilities achieved  in the matching phase. 
In other words, after choosing a channel $m$, if the received reward is non-zero, AP $k$ assumes that this reward is equal to: 
\begin{equation} \label{eq:avReward}
u_k(m)= \hat{\mu}_M(k,m,k_m).
\end{equation}    

The dynamics of this matching phase, adopted from \cite{pareto}, are described in section \ref{subsec:matching}. The matching phase runs for $c_1 l^{1+\delta}$ frames, \revised{where $c_1$ and $\delta$ are constants and $l$ is the epoch number}. 
 The third and final phase is an exploitation phase in which  APs settle on the channels that resulted in the best performance in the previous matching phase. The exploitation phase runs for $c_2 2^l$ timeslots, \revised{$c_2$ being a constant}.

\begin{algorithm}[!ht]
	\caption{}\small 
	\begin{algorithmic}[1]
		\STATEx \textbf{Initialization:} Set $\hat{\mu}_M(k,m, k_m)=0, ~  \forall k\in \mathcal{K}, m \in \mathcal{M}, k_m \in [\beta]$. Set $b^t_k=0, ~ \forall k \in \mathcal{K}$. Let $\epsilon >0$ and $c\geq KN$.
		\FOR  {$l=1,\ldots,L_C$}
		\STATEx \textbf{1- Exploration Phase:} \FOR {$t=1:T_C^0$}
		\STATE Choose one channel $m \in \mathcal{M}$ uniformly.
		\STATE \begin{varwidth}[t]{\linewidth}Receive the achieved reward $x_{k}^t(m)$, and the total number of APs, $k_m$, accessing channel $m$ simultaneously.\end{varwidth}
		\STATE  \begin{varwidth}[t]{\linewidth}$W^t_k(m,k_m)=W^{t-1}_k(m,k_m)+x_{k}^t(m)$,\\
			$co_k^t(k_m)=co_k^t(k_m)+1$.\end{varwidth}
		\IF {$k_m>1$}
		\STATE  $b^t_k=b^{t-1}_k+1$\label{stepAlg:14}
		\ENDIF
		\ENDFOR
		\STATE \begin{varwidth}[t]{\linewidth} Estimate means: $\hat{\mu}_M(k,m,k_m)=\frac{W^t_k(m,k_m)}{co_k^t(k_m)}, \forall ~ k_m \in [\beta]$. \end{varwidth}
		\STATE\begin{varwidth}[t]{\linewidth}  Estimate the number of APs according to: $\hat{K}=\min\left\{\textrm{round}\left(\frac{\log\left(\frac{T_C^0-b^t_k}{T_C^0}\right)}{\log\left(1-\frac{1}{M}\right)}+1\right),\beta M\right\}$. \end{varwidth}
		\STATEx\begin{varwidth}[t]{\linewidth} \textbf{2- Matching Phase:} for the next $c_1l^{1+\delta}$ frames, play according to the dynamics described in section \ref{subsec:matching}. \end{varwidth}
		\STATE\begin{varwidth}[t]{\linewidth}  If $S_k=C$,  choose the action to play according to Eq.\ (\ref{contentChoose}). If $S_k=D$, choose the action according to Eq.\ (\ref{discontentChoose}). \end{varwidth}
		\STATE \begin{varwidth}[t]{\linewidth}  If the achieved reward for some chosen channel ${u_{k}(m)}$, found from Eq.\ (\ref{eq:avReward}), is 0, the AP becomes discontent as per Eq.\ (\ref{zeroTransition}). \end{varwidth}
		\STATE \begin{varwidth}[t]{\linewidth}   If $\boldsymbol{a}_k\neq \boldsymbol{\bar{a}}_k$ or $\boldsymbol{u}_k\neq \boldsymbol{\bar{u}}_k$ or  player $k$ is discontent, the state transition happen according to Eq.\ (\ref{cdTransition}). \end{varwidth}
		\STATE\begin{varwidth}[t]{\linewidth}   Each AP keeps a counter of the number of times each action $\boldsymbol{a_k'}$ was played and resulted in it being content: 
		\begin{equation}
		F_{k}^l(\boldsymbol{a_k'})=\sum_{t=1}^{c_2l^{1+\delta}}\mathbb{I}\left(\boldsymbol{a}_k^t=\boldsymbol{a_k'}, S_k^t=C\right), 
		\end{equation}
		with $\mathbb{I}$ being the indicator function.\end{varwidth}
		\STATEx \textbf{3- Exploitation phase:} for $c_2 2^l$ timeslots:
		\STATE Play the action $\boldsymbol{a_k^{l*}}=\argmax\limits_{\boldsymbol{a}_k\in \mathcal{A}_k} F_{k}^l(\boldsymbol{a}_k)$.
		\ENDFOR
	\end{algorithmic}
	\label{alg1}
\end{algorithm}
\normalsize
\subsection{Matching Dynamics}\label{subsec:matching}
Each AP $k$ is associated with a state $[\boldsymbol{\bar{a}}_k, \boldsymbol{\bar{u}}_k, S]$. The baseline action of AP $k$ is $\boldsymbol{\bar{a}}_k \in \{0,1\}^{M\times 1}$, such that $\sum_{m=1}^M\bar{a}_{k}(m)=N$. The baseline utility of AP $k$ is $\boldsymbol{\bar{u}}_k, \textrm{such that} ~ |\boldsymbol{\bar{u}}_k|=N$. Variable $S \in \{C,D\}$ is the mood of AP $k$ and reflects whether $k$ is content or discontent with the current action and utility. At each frame of the matching phase, each AP chooses an action according to the game dynamics and receives a reward that depends on the collective choices of all the APs. Define $u_{k,\textrm{max}}=\argmax\limits_{\boldsymbol{a}}\sum_{m=1}^Ma_{k}(m)\mu_M(k,m,k_m)$, where $u_{k,\textrm{max}}$ is the highest reward achievable by AP $k$, with a number of estimated APs given by  $\hat{K}$. 

At each frame $t$ during the matching phase, AP $k$ adheres by the following dynamics to decide on the action to choose:  
\begin{itemize}
	\item A content AP plays its baseline action with high probability: 
	\begin{equation}\label{contentChoose}
	p_k^{\boldsymbol{a}_k}=
	\begin{cases}
	\frac{\epsilon^c}{|\mathcal{A}_k|-1}, \quad &\textrm{if} ~  \boldsymbol{a}_k\neq \boldsymbol{\bar{a}}_k, \\
	1-\epsilon^c, \quad &\textrm{if} ~  \boldsymbol{a}_k= \boldsymbol{\bar{a}}_k.,
	\end{cases}
	\end{equation}
	where $\epsilon>0$ is a small perturbation and $c$ is a constant satisfying $c\geq KN$.%
	\item A discontent AP chooses its action uniformly at random:
	\begin{equation}\label{discontentChoose}
	p_k^{\boldsymbol{a}_k}=\frac{1}{|\mathcal{A}_k|}, \quad \forall ~  \boldsymbol{a}_k \in \mathcal{A}_k.
	\end{equation}
\end{itemize} 
\revised{In Eq.\ (\ref{contentChoose}) and (\ref{discontentChoose}), $p_k^{\boldsymbol{a}_k}$ is the probability with which AP $k$ chooses action $\boldsymbol{a}_k$.}

After deciding on the action and observing the reward $u_k(m)$ for  chosen channels, the state transition of each AP $k$ occurs according to: 
\begin{itemize}
	\item If $\boldsymbol{a}_k= \boldsymbol{\bar{a}}_k$ and $\boldsymbol{u}_k= \boldsymbol{\bar{u}}_k$,  a content AP remains content: 
	\begin{equation}\label{contentTransition}
	[\boldsymbol{\bar{a}}_k, \boldsymbol{\bar{u}}_k, C] \rightarrow  [\boldsymbol{\bar{a}}_k, \boldsymbol{\bar{u}}_k, C].
	\end{equation}
	\item If $u_k(m)=0$ for some $m=1,\ldots,N$,  AP $k$ becomes discontent with probability one.
	\begin{equation}\label{zeroTransition}
	[\boldsymbol{\bar{a}}_k, \boldsymbol{\bar{u}}_k, C/D] \rightarrow  [\boldsymbol{a}_k, \boldsymbol{u}_k, D].
	\end{equation}
	\item If $\boldsymbol{a}_k\neq \boldsymbol{\bar{a}}_k$ or $\boldsymbol{u}_k\neq \boldsymbol{\bar{u}}_k$ or  player $k$ is discontent,  the state transitions occur according to: 
	\begin{equation}\label{cdTransition}
	[\boldsymbol{\bar{a}}_k, \boldsymbol{\bar{u}}_k, C/D] \rightarrow
	\begin{cases} 
	[\boldsymbol{a}_k, \boldsymbol{u}_k, C] \quad \textrm{w.p.}~ \epsilon^{u_{k,\textrm{max}}-\sum\limits_{n=1}^N u_{k,n}}, \\
	[\boldsymbol{a}_k, \boldsymbol{u}_k, D] \quad \textrm{w.p.}~  1-\epsilon^{u_{k,\textrm{max}}-\sum\limits_{n=1}^N u_{k,n}}.
	\end{cases}
	\end{equation}
\end{itemize} 

\subsection{Proposed Solution for the Distributed Power Allocation}
A simplified version of Algorithm \ref{alg1} can be used to solve the power allocation problem over each channel $m$. The solution is divided into three phases:
\begin{enumerate}
	\item Exploration phase: This phase runs for $T^0_P$ timeslots and aims at estimating the reward of each power value. During this phase, each AP chooses each of its possible power levels, i.e., power levels in $\mathcal{P}_k^a$, uniformly at random. Upon termination, APs have estimates of the reward associated to each power value, denoted by $\boldsymbol{\hat\mu}_P$.
	\item Matching phase:  In this phase, APs play a non-cooperative game according to the dynamics presented in Section \ref{subsec:matching}, after replacing $\mathcal{A}_k$ in Eq.\ (\ref{contentChoose}) and (\ref{discontentChoose}) by $\mathcal{P}^a_{k,m}$. Each AP keeps a counter of the number of times each action was played and resulted in content behavior.
	\item Exploitation phase: During this phase, each AP $k$ exploits the action, i.e., the power level, that resulted in the most content behavior during the matching phase. 
\end{enumerate}

\section{Regret Analysis}\label{sec:regret}
The time horizon of the channel allocation phase can be lower bounded by \cite{got}:
\begin{equation}
T_C\geq\sum_{l=1}^{L_C-1}(T_C^0+c_1l^{1+\delta}+c_2 2^{l})\geq c_2(2^{L_C}-2),
\end{equation}
where $L_C$ is the total number of epochs occurring within $T_C$ and upper bounded by:
\begin{equation}
L_C\leq \log\left({T_C}/{c_2}+2\right).
\end{equation}
Similarly, the number of epochs, $L_P$, occurring within the time horizon $T_P$ dedicated to the power allocation stage is upper bounded by
$L_P\leq \log(T_P/c_2+2).$
\subsection{Regret in the Exploration Phase}
In the exploration phase of the channel allocation, each AP samples channels uniformly to get estimates of their means. Even though the purpose of this work is to assign to each AP $N$ channels at each timeslot, the number of channels sampled by each AP at a timeslot is set to one in the exploration phase. 
The expected regret incurred by all APs in the exploration phase of the channel allocation, $R_C^1$, can be upper bounded by:
\begin{equation}
R_C^1\leq \sum_{l=1}^{L_C}KNT_C^0\leq KNT_C^0  \log\left({T_C}/{c_2}+2\right).
\end{equation}
Similarly, the expected regret incurred by all APs in the exploration phase of the power allocation, $R_P^1$, can be upper bounded by:
\begin{equation}
R_P^1\leq\sum\limits_{m=1}^M\sum\limits_{l=1}^{L_P}K_m T_P^0\leq K T_P^0\log(T_P/c_2+2).
\end{equation}
\subsection{Regret in the Matching Phase}
The expected regret in the matching phase of the channel allocation, $R_C^2$, can be upper bounded by:
\begin{equation}
R_C^2\leq \sum_{l=1}^{L_C}KNc_1l^{1+\delta}\leq 
KNc_1\log^{2+\delta}\left({T_C}/{c_2}+2\right).
\end{equation}
Similarly, the expected regret in the matching phase of the power allocation, $R_P^2$, can be upper bounded by:
\begin{equation}
R_P^2\leq \sum_{m=1}^{M}\sum_{l=1}^{L_C}K_mc_1l^{1+\delta}\leq 
Kc_1\log^{2+\delta}\left({T_P}/{c_2}+2\right).
\end{equation}
\subsection{Regret in the Exploitation Phase}
In the exploitation phase of epoch $l$ of the channel allocation, each AP $k$ plays the action that it played the most and resulted in content behavior in the matching phase of epoch $l$. The exploitation phase fails in two cases:
\begin{enumerate}
	\item If the exploration phase of epoch $l$ fails: This happens with a probability $\leq 4(M\beta)^2 e^{-l}$ as shown in Lemma \ref{lemma2}.
	\item If the most played action of the matching epoch differs from the optimal action: This happens with a probability $\leq A_1 e^{-l^{1+\delta}}$ as shown in Lemma \ref{lemma6}.
\end{enumerate}
The expected regret incurred by all APs in the exploitation phase can be upper bounded by:
\begin{equation}
\begin{aligned}
R_C^3&\leq \sum_{l=1}^{L_C}KNc_22^l\left(4(M\beta)^2 e^{-l}+A_1 e^{-l^{1+\delta} }\right)\leq A_3,
\end{aligned}
\end{equation}
\revised{where $A_1, A_3$ are constants.}

Similarly, the regret incurred by the APs in the exploitation phase of the power allocation is $R_P^3\leq A_3$.
\subsection{Regret of the Proposed Technique}
\begin{theorem}
	The expected regret of the proposed allocation solution can be upper bounded as:
	\begin{equation}
	R\leq R_C^1+ R_C^2+R_C^3+R_P^1+ R_P^2+R_P^3=\mathcal{O}\left(\log^{2+\delta}(T)\right).
	\end{equation}
\end{theorem}

\section{Exploration Phase}\label{sec:expPhase}
The exploration phase is performed so APs learn estimates of the channel mean reward in the channel allocation phase, and of the power level mean reward in the power allocation phase. Moreover, by keeping track of the number of times each channel was accessed with one or more other APs in the channel allocation phase, the APs can estimate the total number of APs in the system. In this section, we  find the minimum  length of the exploration phase ensuring an accurate estimation of  both the reward means and the number of APs.

\subsection{Estimation of the Reward Means }

Since the estimation may not always be perfect, the result of the assignment with the estimated means ($\boldsymbol{\hat{\mu}}_M$ and $\boldsymbol{\hat{\mu}}_P$) might differ from the result of the assignment calculated with the true means ($\boldsymbol{{\mu}}_M$ and $\boldsymbol{{\mu}}_P$). However, if the estimation inaccuracy is kept small as in \cite{magesh} and \cite{got}, the result of the assignment would not be affected. 

\begin{lemma}\label{lemma1} Let $J_M^1$ and $J_M^2$ be the sum reward achieved by the best channel assignment and the second best channel assignment and let $\Delta_M=\frac{J_M^1-J_M^2}{2KN}$. \vspace{0.04cm} Moreover, let $J_P^1$ and $J_P^2$ be the sum reward achieved by the best power allocation on each channel $m$ and the second best power assignment and let $\Delta_P=\frac{J_P^1-J_P^2}{2K_m}$. If the difference between the estimated and the correct  reward means satisfies:
	\begin{equation}
	\begin{aligned}\label{deltaC}
	|\mu_M(k,m, k_m)-\hat{\mu}_M(k,m, k_m)|< \Delta_M, 
	\forall k\in\mathcal{K}, m\in\mathcal{M}, k_m\in[\beta],
	\end{aligned}
	\end{equation}
	\begin{equation}\label{deltaP}
	\begin{aligned}
	|\mu_P(k,m,v_l)-\hat{\mu}_P(k,m,v_l)|< \Delta_P, \forall k\in\mathcal{K}_m, m\in\mathcal{M}, v_l\in\mathcal{VL},
	\end{aligned}
	\end{equation}  then, the best assignment result does not change due to the estimation inaccuracy.
\end{lemma}

\begin{proof}
	See Appendix\ \ref{appendix1}.
\end{proof}
%
Next, we upper bound the probability of error, i.e., the probability of having channel reward estimates (resp. power level reward estimates) that do not satisfy the condition in\ (\ref{deltaC}) (resp. (\ref{deltaP})) in the exploration epoch $l$. We also provide  a lower bound of the length of the exploration epoch $T_{\boldsymbol{\hat{\mu}_M}}$ in the channel allocation phase, and $T_P^0$ in the power allocation phase.

\begin{lemma}\label{lemma2}
	If $T_{\boldsymbol{\hat{\mu}_M}}=
	\left\lceil{\frac{2Me^{\left(\frac{K-1}{M-1}\right)}}{\Delta_M^2\left(M-1\right)^{1-\beta}}}\right\rceil$, \vspace{0.05cm} all players have an estimate of the channel means satisfying the condition in\ (\ref{deltaC}), with probability $\geq 1-\gamma_{e,l}^M$, \revised{ where $\gamma^M_{e,l}$ is  the probability of error in the $l^{\textrm{th}}$ exploration phase of the uncoordinated channel access.} Moreover, $ \gamma^M_{e,l}\leq 4(M\beta)^2 e^{-l}$. 
	
	For the power allocation exploration phase, if $T_P^0=
	\left\lceil{\frac{2Le^{\left(\frac{\beta-1}{L-1}\right)}}{\Delta_p^2}}\right\rceil$, all players have an estimate of the power level means satisfying the condition in\ (\ref{deltaP}),  with probability $\geq 1-\gamma^P_{e,l}$, \revised{where $\gamma^P_{e,l}$,  is  the probability of error in the $l^{\textrm{th}}$ exploration phase of the power allocation, } upper bounded by $ 4\beta L e^{-l}$.
\end{lemma}
\begin{proof}
	See Appendix\ \ref{appendix2}.
\end{proof}
%
%
We now turn our attention to finding the minimum length of the exploration phase in the channel allocation stage ensuring an accurate estimate of the number of APs $\hat{K}$.
\subsection{Estimating the number of APs}
For AP $k$, $b_k^t$ found in step\ \ref{stepAlg:14} of Algorithm\ \ref{alg1} denotes the number of timeslots player $k$ was not the sole occupier of some channel $m$ until  $t$. 
\begin{lemma}\label{lemma3}
	If the length of the exploration epoch in the channel allocation step satisfies:
	\begin{equation}\label{eq:t_kChannel}
	T_{\hat{K}}=
	\left\lceil2.08\log{\left(\frac{2}{\eta}\right)}M^2e^{2\left(\frac{M\beta-1}{M-1}\right)} \right\rceil,
	\end{equation}
	then all APs have an estimate of the number of APs $\hat{K}$ satisfying $\hat{K}=K$ with probability higher than $1-\eta$, \revised{where $\eta$ is the probability of error in the estimation of the number of APs.}
\end{lemma}
\begin{proof}
	See Appendix\ \ref{appendix3}.
\end{proof}
\subsection{ Length of the Channel Allocation Exploration Phase}
To ensure an accurate estimate of the channel reward means and of the number of APs, the minimum length of the exploration phase in the channel allocation solution, $T_C^0$, must satisfy the conditions in Lemma \ref{lemma2} and Lemma \ref{lemma3}. Hence, the following must hold:
\begin{equation}
\begin{aligned}
T_C^0=\max&\left\{\left\lceil\frac{2Me^{\left(\frac{K-1}{M-1}\right)}}{\Delta_M^2\left(M-1\right)^{1-\beta}}\right\rceil, \right. \left. \left\lceil 2.08\log{\left(\frac{2}{\eta}\right)}M^2e^{2\left(\frac{M\beta-1}{M-1}\right)} \right\rceil\right\}.
\end{aligned}
\end{equation}
\section{Matching Phase}\label{sec:matPhase}
The matching phase of the channel allocation solution aims at reaching a final assignment in which every AP accesses $N$ channels, such that the achieved sum reward  is maximized. 
%

The dynamics presented in section \ref{subsec:matching} and adopted in the matching phase induce a Markov chain over the state space $\mathcal{Z}=\prod_{k=1}^K\{\mathcal{A}_K\times [0,1]^{N\times 1} \times \{C,D\}\}$. Let $P^\epsilon $ denote the transision matrix of the regular perturbed Markov chain $\mathcal{Z}$. The work in \cite{pareto} guarantees that, when playing according to these dynamics, the optimal state, i.e., the one maximizing the sum rewards, is played most often. The proof relies on the theory of resistance trees for regular perturbed Markov chains \cite{evolution}. The dynamics used in this paper differ from those in \cite{pareto} in two aspects:
\begin{enumerate}
	\item	If AP $k$ receives a reward equal to $0$ on some channel $m$,  AP $k$ is discontent with probability one. In \cite{pareto}, the game is assumed to be interdependent which means that it is not possible to  partition APs into two groups that do not interact with each other. However, this property does not hold in the considered setting as shown in \cite{got}.   Therefore, as in \cite{got}, to characterize the stable states of the unperturbed chain when $\epsilon=0$, a player with 0 reward  on some channels is discontent with probability one. 
	\item For the transition probabilities between content and discontent in (\ref{cdTransition}), instead of using $\epsilon^{N-\sum\limits_{n=1}^Nu_{k,n}}$, we use $\epsilon^{u_{k,\textrm{max}}-\sum\limits_{n=1}^Nu_{k,n}}$, since the maximum utility achievable by each AP $k$ is $u_{k,\textrm{max}}$.
\end{enumerate}
Next, the recurrence states of $\mathcal{Z}$ are characterized. 
\begin{lemma}\label{lemma5}
	Let $D^0$ denote the set of states where all APs are discontent. Moreover, let $C^0$ denote all singleton states where all APs are content and their baseline actions and utilities are aligned. As proved in \cite{pareto}, the only recurrence states of $\mathcal{Z}$ are $D^0$ and all singletons in $C^0$.
\end{lemma}



The resistance of moving from one recurrence state to the other being similar to \cite{pareto}, the stochastic potential of any state $z\in C^0$ is of the form:
\begin{equation}
\zeta(z)=c[|C^0|-1]+\sum\limits_{k=1}^K u_{k,\text{max}}-\sum\limits_{m=1}^Ma_k(m)\hat{\mu}(k,m,k_m).
\end{equation}
 From Theorem 1 of \cite{pareto}, the stable state is the one minimizing the stochastic potential, hence the one maximizing the achieved sum reward.  This stable state is guaranteed to be played the majority of times for a small enough perturbation $\epsilon$ \cite{got}, \cite{pareto}. In the exploitation phase, as the state that was most played and that resulted most in the players being content is played, the stable state is hence expected to be played with high probability. Next, the probability of error in the matching epoch $l$ is found.

Let $\pi$ denote the stationary distribution of the Markov chain $\mathcal{Z}$ and let $\boldsymbol{z^*}=[\boldsymbol{\bar{a}^*}, \boldsymbol{\bar{u}^*}, C^K]$ denote the optimal state. According to \cite{got}, $\pi(\boldsymbol{z^*})>1/2$ for a small enough perturbation $\epsilon$. The following lemma finds the probability of error in the matching phase of the $l^\text{th}$ epoch, $\delta_{m,l}$.  
\begin{lemma}\label{lemma6}Let $\boldsymbol{a}^{(l)}$ denote the action that was most played in some epoch $l$. As proved in \cite{magesh}, the probability of error in the matching phase in epoch $l$, $\delta_{m,l}$, is upper bounded by:
	\begin{equation}
	\delta_{m,l}=\text{Pr}(\boldsymbol{a^*}\neq \boldsymbol{a}^{(l)})\leq A_0\norm{\phi}_{\pi}\exp\left(\frac{-\theta^2\pi(\boldsymbol{z^*})c_2l^{1+\delta} }{72T_m(1/8)}\right),
	\end{equation}   
	where $A_0$ is a constant, $\phi_{\pi}$ is the probability distribution of the initial state played in epoch $l$ and $T_m(1/8)$ is the mixing time of the Markov chain $\mathcal{Z}$ with an accuracy of 1/8 \cite{chernoff}. 
\end{lemma}
The analysis of the matching phase of the power allocation solution is similar to the one given above and is omitted for space constraints.
\section{Simulation Results}\label{sec:results}
Extensive simulations of the proposed algorithm were conducted to validate its performance. The following simulation parameters were chosen: $K=4, M=4, N=\beta=L=2, B_c= 2.5 ~\textrm{MHz}, c_1=3000, c_2=5000, \epsilon = 5 \times 10^{-5}, \gamma=0$. The available SINR values are $\boldsymbol\Gamma=\{24, 4.77\}\, \text{(dB)}$ leading to achieved rates of 20 and 5 Mbps respectively. For the channel allocation stage, the parameter $c$ used in the matching phase (Cf. Section \ref{subsec:matching}) is set as: $ c = KN$, whereas for the power allocation stage $c= K_m$ for each channel $m\in \mathcal{M}$. {Two of the APs are assumed to have a power budget of 1W per channel, while the remaining two have a power budget of 2W per channel.} Additional simulation parameters are given in Table\ \ref{tab1} \cite{3gpp}.
{\renewcommand{\arraystretch}{1.1}
	\begin{table}[!htb]
		\small
		\caption{Simulation parameters.}
		\centering
		\vspace{-0.2cm}
		\begin{tabular}{|c|c|}
			\hline Cell Radius $R_d$ & 150 m\\
			\hline Overall Transmission Bandwidth & 10 MHz\\
			\hline Number of channels & 4\\
			\hline Number of APs & 4\\
			\hline Power Budget per AP & \multirow{2}{*}{$\{1,1,2,2\}$ (W)}\\
			per channel $P_{(.)}^m$ &\\
			\hline Available SINR Requirements & $\boldsymbol{\Gamma}=\{24, 4.77\} \text{(dB)}$  \\
			\hline  \multirow{2}{*}{Distance Dependent Path Loss} & $128.1 + 37.6 \log_{10}(d) \text{(dB)},$ \\
			&$d \:\text{in Km}$ \\
			\hline {Receiver Noise Density} &  {$4.10^{-18}$ mW/Hz}\\
			\hline		
		\end{tabular}\label{tab1}\vspace{-0.2cm}
\end{table}}
\normalsize
\vspace{-0.4cm}
\subsection{Estimation Accuracy of the Exploration Phase}\label{subsec:estimation}
\vspace{-0.6cm}
\begin{figure}[!h]
	\begin{center}
		\begin{subfigure}{.33\columnwidth}
			\centering
			\includegraphics[height=4.5cm,keepaspectratio]{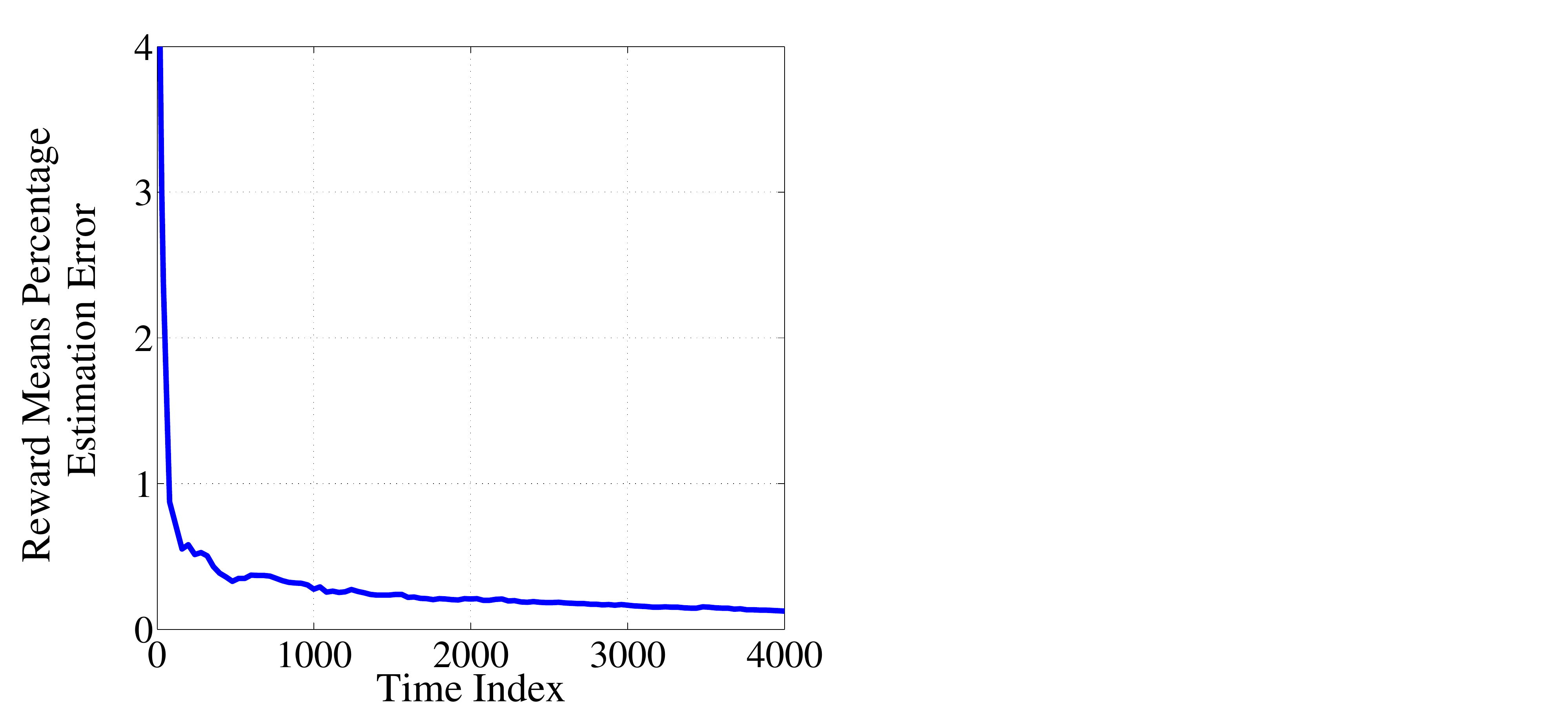}  
			\caption{\label{fig:fig0a} }
		\end{subfigure}%
		\begin{subfigure}{0.33\columnwidth}
			\centering
			\includegraphics[height=4.5cm,keepaspectratio]{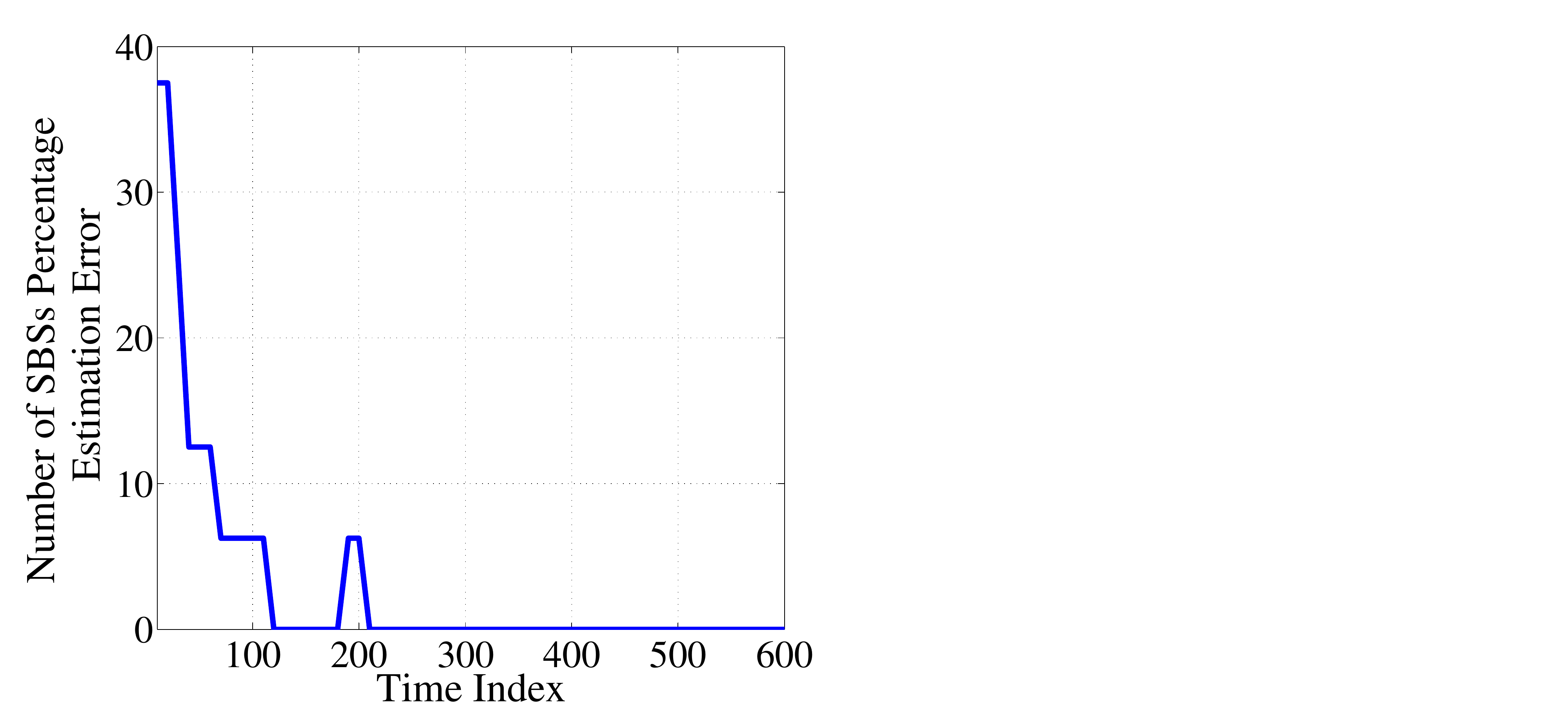}  
			\caption{\label{fig:fig0b} }
		\end{subfigure}
		\begin{subfigure}{0.33\columnwidth}
		\centering
		\includegraphics[height=4.5cm,keepaspectratio]{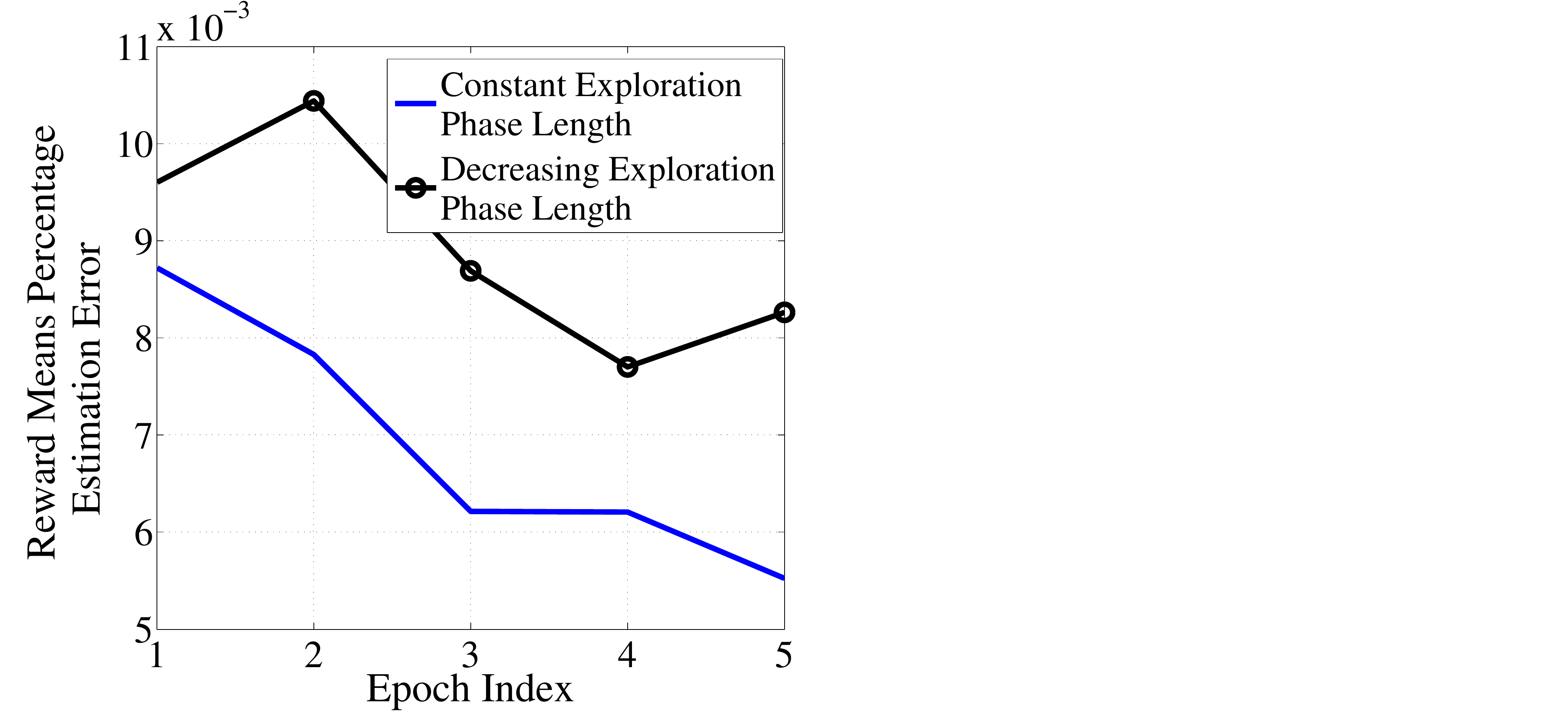}  
		\caption{\label{fig:fig0c} }
	\end{subfigure}
\vspace{-0.2cm}
		\caption{\label{fig:fig0} {Estimation error as time progresses in the channel allocation stage for (a) the estimation of the rewards, (b) the estimation of the number of APs. (c) Comparison of the estimation error as a function of the epoch index in the channel allocation stage for  the estimation of the rewards.}}
	\end{center}
\end{figure}
\vspace{-0.8cm}
First, we evaluate the estimation accuracy of the exploration phase in the channel allocation stage. As shown in Fig.\ \ref{fig:fig0a} and Fig.\ \ref{fig:fig0b}, the estimation of both the reward means and the total number of APs converges rather quickly to the correct values. 
Having observed that the estimation of the exploration phase converges quickly, a version of the proposed algorithm where the exploration phase length is divided by the epoch index  was tested. The estimation error of this version with a decreasing exploration phase length was compared against the version with a constant exploration phase length. Fig.\ \ref{fig:fig0c} plots the channel rewards estimation error for both versions. Although the constant length version outperforms the version with a decreasing exploration phase length, the estimation error achieved by both methods is lower than $1.1\times 10^{-2} \%$, hence negligible.  When it comes to the number of APs estimation, both versions accurately estimate $\hat{K}$, without error, when convergence is reached.

For the power allocation stage, the power level rewards estimation also converges quickly to a negligible error value.
\vspace{-0.4cm}
\subsection{Performance Analysis}
\vspace{-0.6cm}
\begin{figure}[!h]
	\begin{center}
		\begin{subfigure}{.33\columnwidth}
			\centering
	 \includegraphics[height=4.5cm,keepaspectratio]{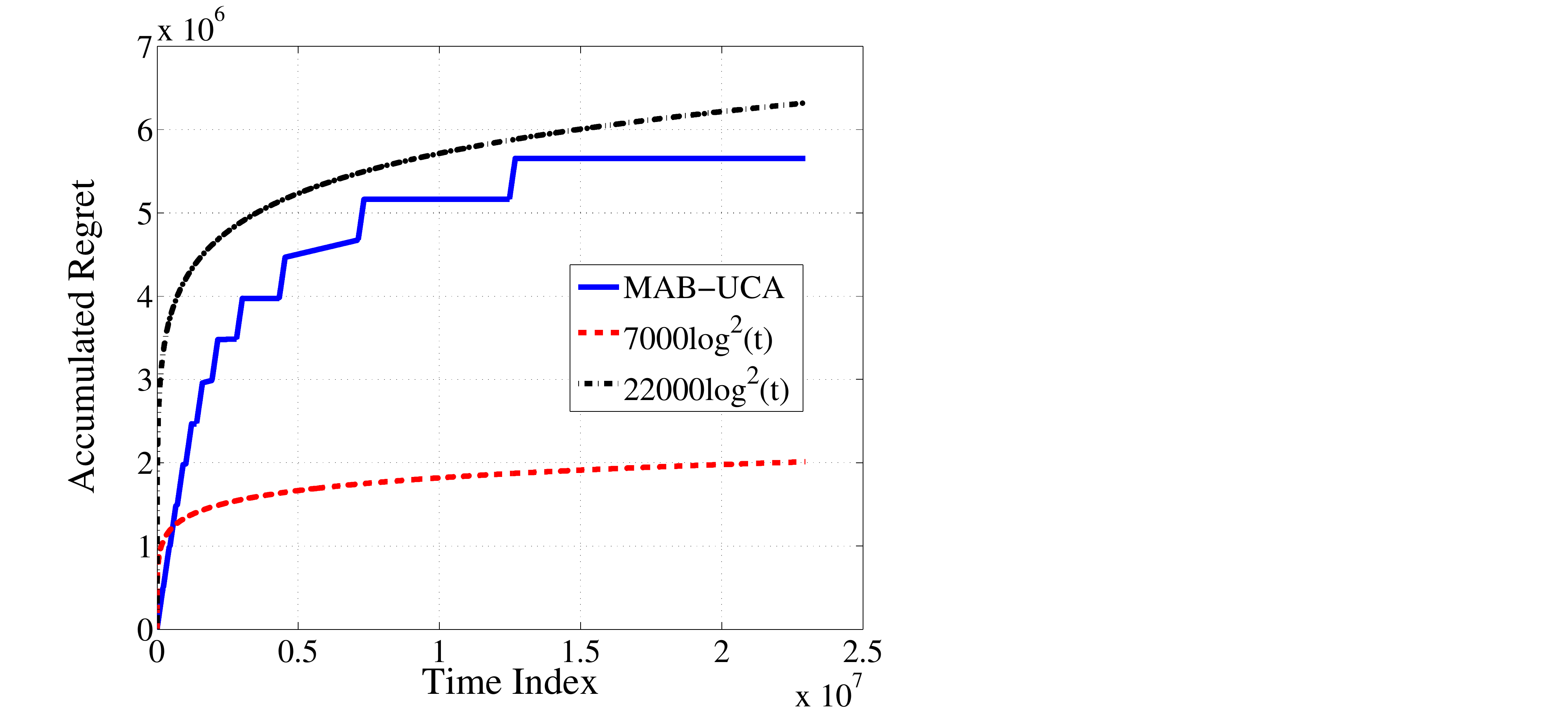}  
			\caption{\label{fig:fig1a} }
		\end{subfigure}%
		\begin{subfigure}{0.33\columnwidth}
			\centering
		\includegraphics[height=4.5cm,keepaspectratio]{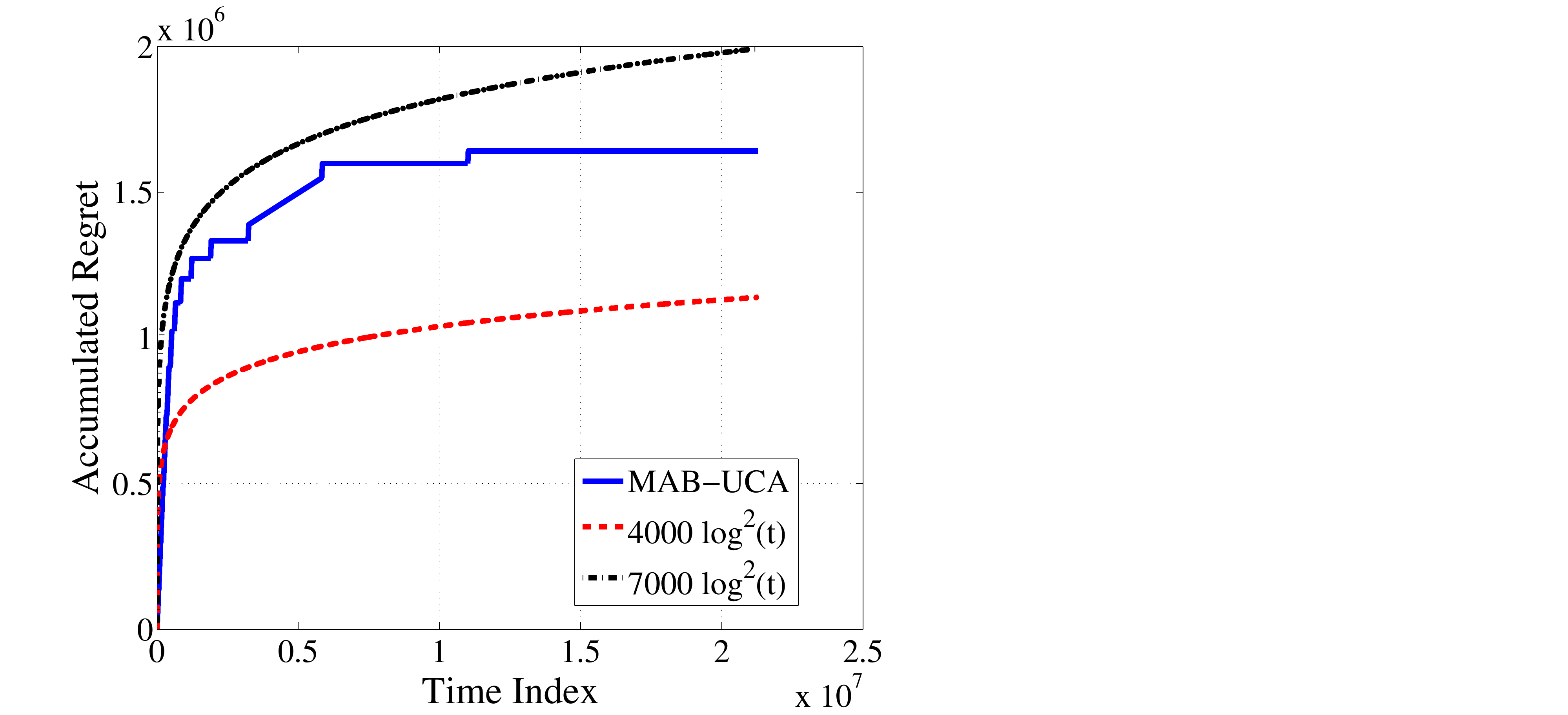}  
			\caption{\label{fig:fig1b} }
		\end{subfigure}
		\begin{subfigure}{0.33\columnwidth}
		\centering
		\includegraphics[height=4.5cm,keepaspectratio]{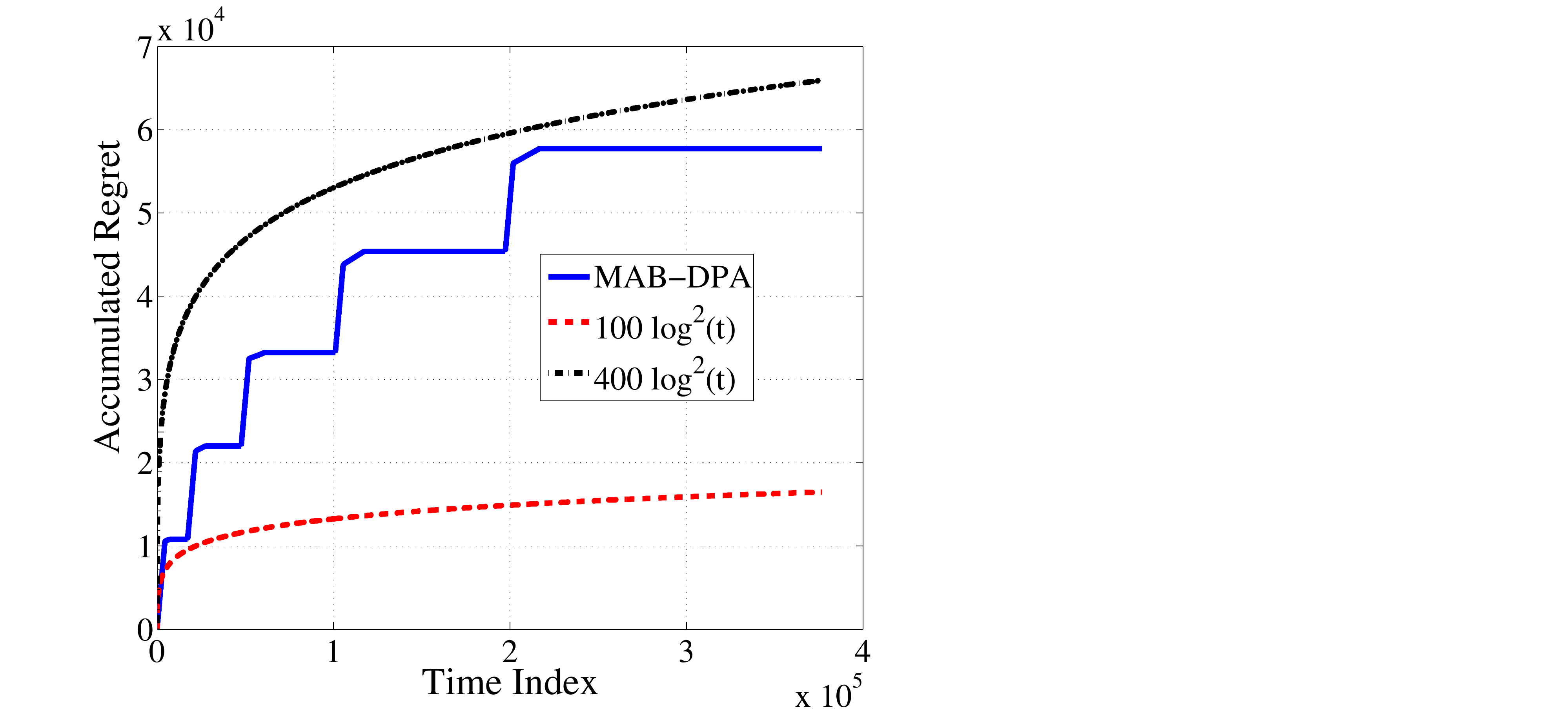}  
		\caption{\label{fig:fig1c} }
	\end{subfigure}
\vspace{-0.2cm}
		\caption{\label{fig:fig1} {Accumulated regret as time progresses  (a) for the channel allocation phase with a constant exploration phase length, (b) for the channel allocation phase with a decreasing exploration phase length, (c) for the power allocation stage.}}
	\end{center}
\end{figure}

\revised{Fig. \ref{fig:fig1} shows the average accumulated regret as a function of time in the channel allocation stage for both the constant and the decreasing length exploration phase versions. The results show that the average accumulated regret for both versions increases with time as $\mathcal{O}(\log(t)^2)$. More specifically, the regret incurred for the constant length exploration phase version  is bounded between $7000 \log(t)^2$ and $22000\log(t)^2$, as shown in Fig.  \ref{fig:fig1a}. The regret incurred for the decreasing length exploration phase version is  bounded between $4000 \log(t)^2$ and $7000\log(t)^2$.  In fact, most of the regret is accumulated during the exploration phase where APs choose a channel uniformly at random. Hence, decreasing the length of the exploration phase lowers the value of the accumulated regret as shown in Fig.\ \ref{fig:fig1b}, without jeopardizing the estimation accuracy as was shown in Section \ref{subsec:estimation}.}

\revised{The regret incurred on all channels during the power allocation stage  is bounded between $100 \log(t)^2$ and $400\log(t)^2$, as shown in Fig.  \ref{fig:fig1c}. The lower regret observed during the power allocation stage, when compared to the channel allocation stage, results from the smaller number of APs competing  for a smaller number of arms.  In fact, on each channel $m \in \mathcal{M}$ during the power allocation stage, the number of competing APs  is $K_m\leq \beta=2$, while the number of arms or power levels is $L=2$. In contrast, during the channel allocation stage, the number of players is $K=4$ with $\binom{M}{N}=6$ available arms. }
\revised{magenta}{\begin{remark}
 To provide insight on the accumulated regret as a function of time in seconds, and the time duration needed to reach convergence, assume that a subcarrier spacing of 240 KHz \cite{etsi} is considered, resulting in a timeslot duration equal to 62.5 $\mu$s. For the uncoordinated channel access part of the solution, convergence to the optimal allocation is first reached at the fourth epoch, which takes place from $0.45\times 10^6$ to $0.6\times 10^6$ timeslots approximately. In terms of time duration in seconds, convergence is reached in $0.45\times 10^6\times 62.5 \times 10^{-6}=28.125$ seconds. For the uncoordinated power control part, convergence is reached from the first epoch,  i.e., at around $0.1 \times 10^5$   timeslots, or 0.625 seconds with a timeslot duration of 62.5 $\mu$s.
\end{remark}}
In Fig.\ \ref{fig:fig2}, we compare the performance of the proposed method with a technique based on the UCB algorithm proposed in \cite{8902878} and similar to the one proposed in \cite{8676344}, denoted by Two-Dimensional UCB. In the Two-Dimensional UCB method, channel and power allocation are conducted at the same time, using the UCB algorithm, by considering all possible combinations of the channels and the power levels. For the considered setting, the number of arms in the Two-Dimensional UCB method is hence $\binom{M}{N}\times L^N=24$ arms.    

\begin{figure}[!h]
	\begin{center}
		\begin{subfigure}{.33\columnwidth}
	\centering
		\includegraphics[height=4.5cm,keepaspectratio]{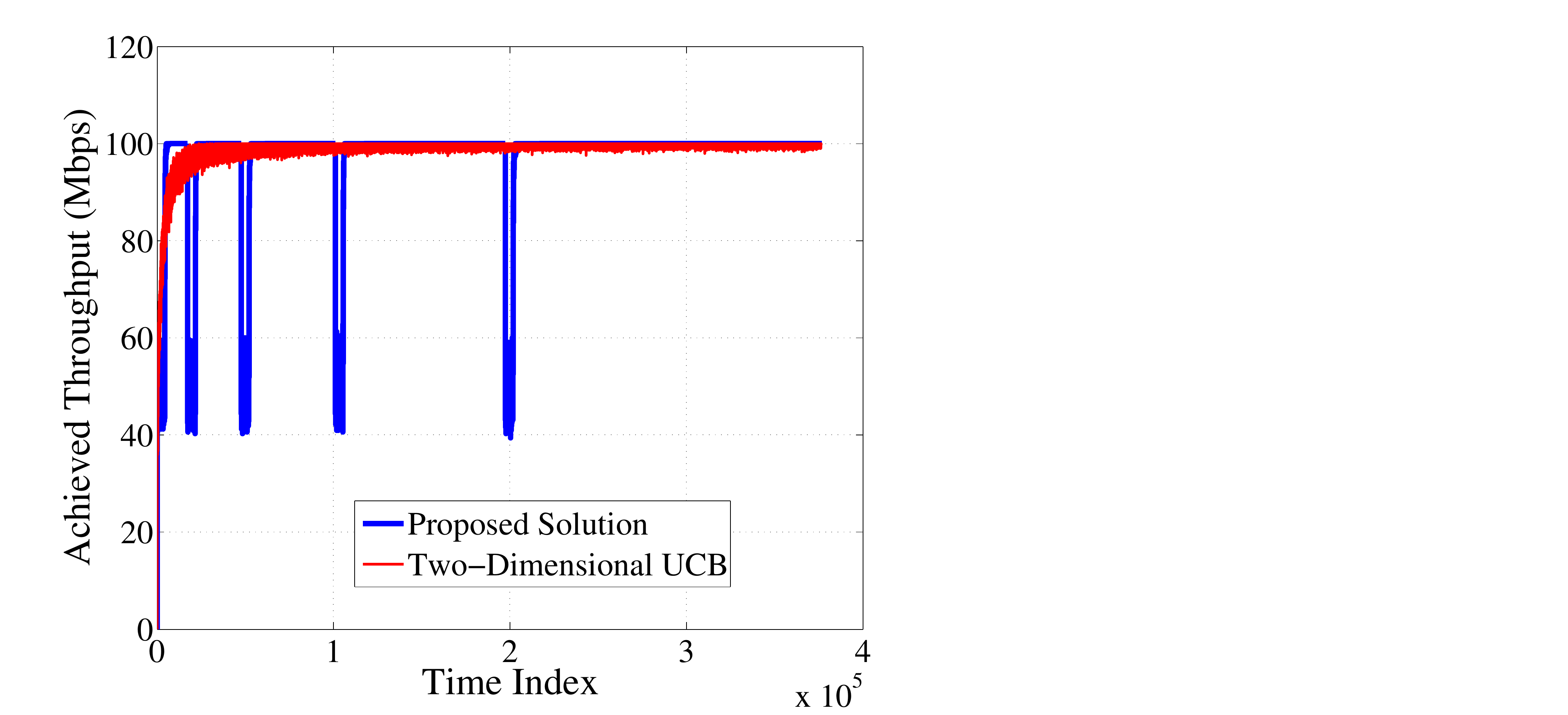}  
			\caption{\label{fig:fig2a} }
		\end{subfigure}%
		\begin{subfigure}{0.33\columnwidth}
			\centering
		\includegraphics[height=4.5cm,keepaspectratio]{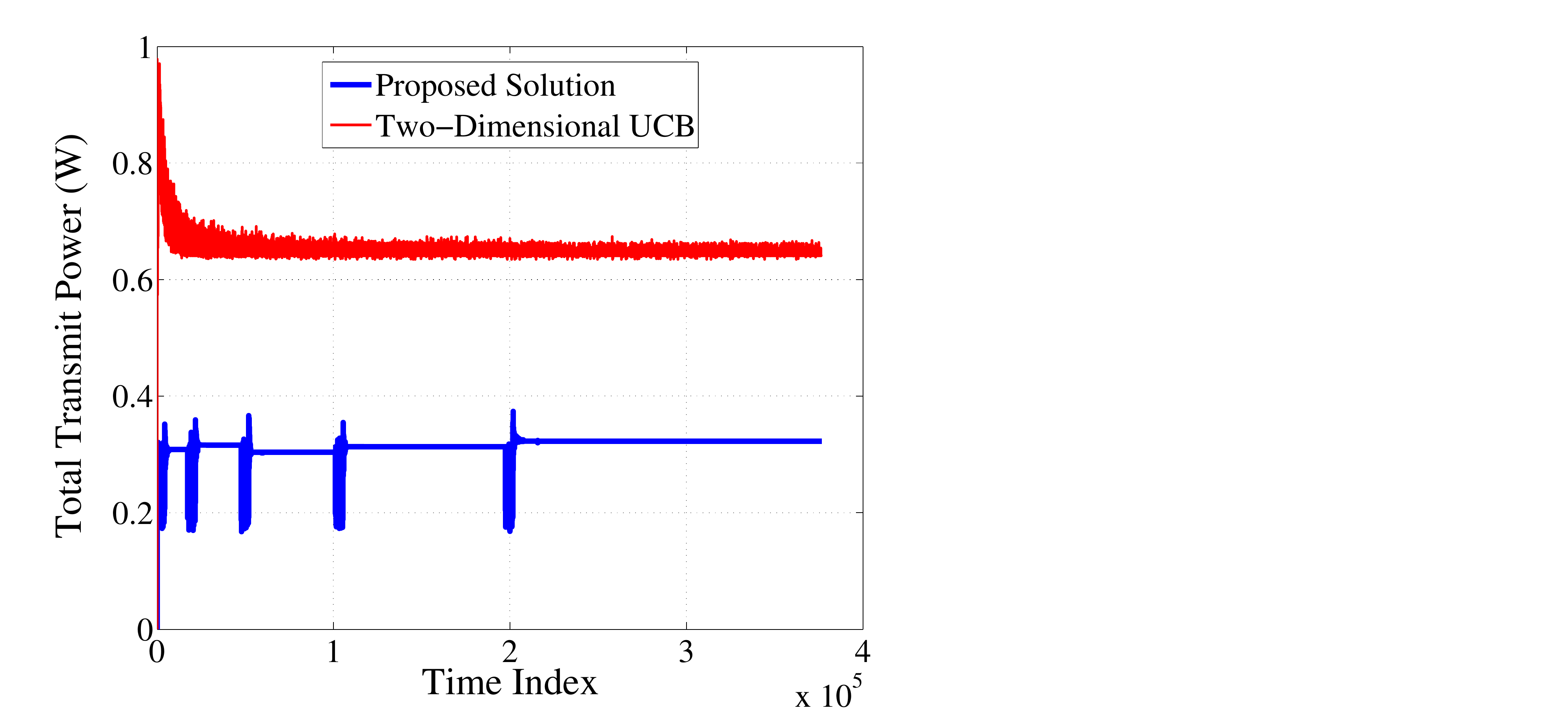}  
			\caption{\label{fig:fig2b} }
		\end{subfigure}
		\begin{subfigure}{0.33\columnwidth}
		\centering
		\includegraphics[height=4.5cm,keepaspectratio]{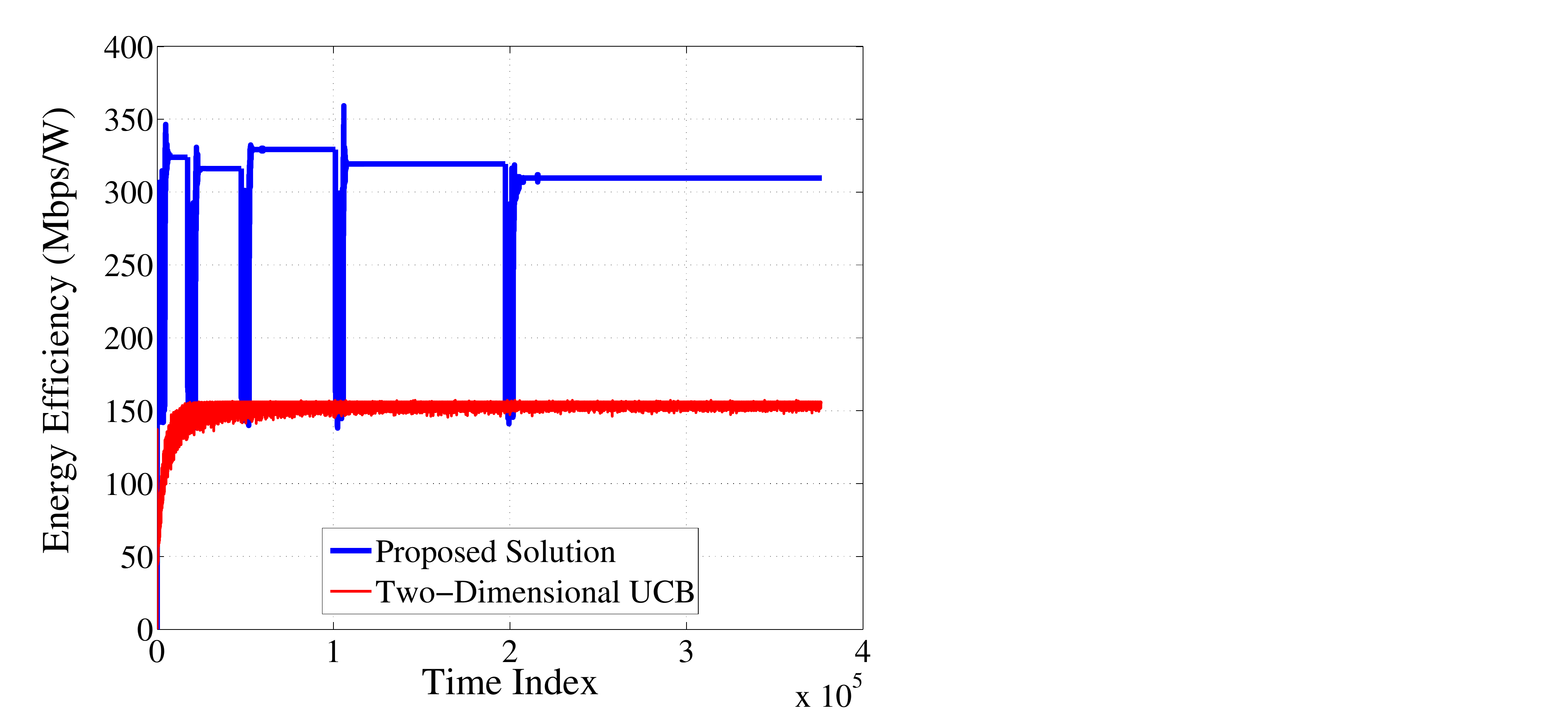}  
		\caption{\label{fig:fig2c} }
	\end{subfigure}
		\vspace{-0.2cm}
		\caption{\label{fig:fig2} { Performance comparison as a function of time of (a) the achieved rate, (b) the total transmit power, (c) energy efficiency.}}
	\end{center}
\end{figure}
\vspace{-0cm}
\revised{In Fig.\ \ref{fig:fig2a}, the achieved rate is plotted as a function of time. Both methods converge relatively quickly to the highest achievable rate, with small variations for the Two-Dimensional UCB technique.  The sharp falls in the achieved rate of the proposed method are due to the exploration phase during each epoch of the power allocation stage where APs choose the power levels uniformly at random, {causing collisions and leading to zero rates}.}

\revised{The total transmit power used by the APs as a function of time is shown in Fig.\ \ref{fig:fig2b}. While both methods converge to the same highest achievable rate, the  power used by our proposed method is significantly lower than the one needed by the Two-Dimensional UCB method. \revised{This means that the UCB-based method does not lead APs to learn the optimal allocation and converges to a sub-optimal resource partitioning among the APs. In other words, our proposed method achieves a better allocation for the channel and power when compared to the UCB-based method.} Moreover, our proposed method has performance guarantees in terms of regret and optimality, while the Two-Dimensional UCB method  \cite{8902878} does not.}

\revised{To check the combined effect of rate and power on the performance of the compared methods, the achieved energy efficiency (EE), which is the ratio of the achieved rate to the used power,  is plotted in Fig.\ \ref{fig:fig2c}. Once again,   the sharp falls in the performance of our proposed method are due to the exploration phase in each epoch of the power allocation stage. \revised{Fig.\ \ref{fig:fig2c} shows that our proposed method greatly outperforms the UCB-based method, by achieving more than a twofold increase in the EE. This is due to our method converging to the optimal allocation when the UCB-based technique converges to a sub-optimal allocation requiring more transmit power as shown by Fig.\ \ref{fig:fig2b}.}}
\section{Conclusion}\label{sec:conc}
In this paper, the uncoordinated channel and power allocation problems in a SON were studied. The considered framework allows each AP to choose $N$ channels at each timeslot, and allows each channel to simultaneously accommodate multiple APs in a NOMA manner. The considered problem  was modeled using the  multi-player MAB  framework, with varying user rewards, multiple plays, and non-zero reward on collision. A game-theoretic approach was used to develop an algorithm with a sub-linear regret  of $\mathcal{O}(\log^2T)$.  Simulation results validated the sub-linear regret of the proposed method and showed its superior performance, when compared with one of the most used algorithms in the MAB literature.
\section*{Acknowledgment}
The authors would like to thank Akshayaa Magesh for useful discussions regarding multi-player multi-armed bandits.
\begin{appendices}
	\section{Proof of Lemma \ref{lemma1}}\label{appendix1}
	In the channel allocation phase, denote by $\boldsymbol{a}^{(1)}$ the optimal assignment, and by $J_M^1$ the sum rewards achieved when $\boldsymbol{a}^{(1)}$ is played, which is then given by:
	\begin{equation}
	J_M^1=\sum_{k=1}^K\sum_{m=1}^M a_{k}^{(1)}(m)\mu_M(k,m,k^*_m).
	\end{equation}
	Furthermore, denote the second best assignment and the sum reward achieved under it by $\boldsymbol{a}^{(2)}$ and $J_M^2$ respectively.   Let the estimated mean of AP $k$ over channel $m$ with $k_m$ APs on channel $m$ be written as: 
	\begin{equation}
	\hat{\mu}_M(k,m,k_m)=\mu_M(k,m,k_m)+z(k,m,k_m),
	\end{equation}
	where $z(k,m,k_m)$ is the estimation inaccuracy during the channel allocation phase satisfying $|z(k,m,k_m)|\leq \Delta_M$. 
	The sum reward achieved when $\boldsymbol{a}^{(1)}$ is played with the estimated channel means satisfies: 
	\begin{equation}
	\begin{aligned}
	\sum_{k=1}^K\sum_{m=1}^M a_{k}^{(1)}(m)\hat{\mu}_M(k,m,k_m)=& \sum_{k=1}^K\sum_{m=1}^M a_{k}^{(1)}(m)(\mu_M(k,m,k_m)+z(k,m,k_m))> \\
	&\sum_{k=1}^K\sum_{m=1}^M a_{k}^{(1)}(m)\mu_M(k,m,k_m)- K N \Delta_M.
	\end{aligned}
	\end{equation}
	Any other assignment $\boldsymbol{a}\neq\boldsymbol{a}^{(1)}\neq \boldsymbol{a}^{(2)}$ must perform at most as well as $\boldsymbol{a}^{(2)}$:
	\begin{equation}
	\begin{aligned}
	\sum_{k=1}^K\sum_{m=1}^M a_{k}(m)\hat{\mu}_M(k,m,k_m)=&\sum_{k=1}^K\sum_{m=1}^M a_{k}(m)(\mu_M(k,m,k_m)+z(k,m,k_m))< \\
	&\sum_{k=1}^K\sum_{m=1}^M a_{k}^{(2)}(m)\mu_M(k,m,k_m)+ K N \Delta_M. 
	\end{aligned}
	\end{equation}
	To avoid changing the optimal assignment because of the estimation inaccuracy, the following  must hold $ \forall \boldsymbol{a}\neq \boldsymbol{a}^{(1)}$:
	\begin{equation}\label{eq:conditiona}
	\begin{aligned}
	\sum_{k=1}^K\sum_{m=1}^M a_{k}^{(1)}(m)\hat{\mu}_M(k,m,k_m)>\sum_{k=1}^K\sum_{m=1}^M a_{k}(m)\hat{\mu}_M(k,m,k_m).
	\end{aligned}
	\end{equation}
	To ensure Eq.\ (\ref{eq:conditiona}), we need to have: $J_M^1 -  KN\Delta_M>J_M^2+ KN\Delta_M$, which holds if:
	\begin{equation}
	\Delta_M<\frac{J_M^1-J_M^2}{2KN}.
	\end{equation}
	In the power allocation phase, following a similar approach over each channel $m$, we get:
	\begin{equation}
	\Delta_P<\frac{J_P^1-J_P^2}{2K_m}.
	\end{equation}
	\section{Proof of Lemma\ \ref{lemma2}}\label{appendix2}
	\subsection{Lower Bound of the Length of the Exploration Phase in the Channel Allocation Step}\label{appendix2:channel}
	To find a lower bound of the length of the exploration phase in the channel allocation step,  we first find the required number of observations of each channel by each AP to guarantee  condition\ (\ref{deltaC}) \cite{meghana,rosenski}. To do so,  the  probability of each AP not having a correct estimation of the channel means should be bounded. Let $\gamma=\gamma_{e,l}^M/2$. Define the following events:
	\begin{itemize}
		\item $A$: all players have an estimate satisfying condition\ (\ref{deltaC}),
		\item $B$: all players have $\geq Q$ ~observations of each channel $m$ for every $s$ in $[\beta]$,
		\item $A_k$: player $k$ has an estimate satisfying condition (\ref{deltaC}),
		\item $B_k$: player $k$ has $\geq Q$ ~observations of each channel $m$ for every $s$ in $[\beta]$.
	\end{itemize} 
	
	The following must hold:
	\begin{equation}
	\textrm{Pr} (\bar{A}_k~|~B_k )\leq \frac{\gamma}{K}.
	\end{equation}
	In fact,
	\begin{equation}
	\begin{aligned}
	&\textrm{Pr} (\bar{A}_k|B_k)\leq \textrm{Pr} ~ (\exists ~m,s, \textrm{s.t.} ~|\mu_M(k,m,s)-\hat{\mu}_M(k,m,s)|> \Delta_M~ |~ B_k) \overset{\text{(a)}}\leq\\
	&\sum_{m=1}^M\sum_{s=1}^{\beta}\textrm{Pr}~ (|\mu_M(k,m,s)-\hat{\mu}_M(k,m,s)|> \Delta_M~ |~ B_k)=\\
	&\sum_{m=1}^M\sum_{s=1}^{\beta}\sum_{q=Q}^\infty\textrm{Pr} ~(|\mu_M(k,m,s)-\hat{\mu}_M(k,m,s)|> \Delta_M~ | k ~\textrm{has}  ~q \textrm{ observations of $(m,s)$}) \times p_2 \overset{\text{(b)}}\leq\\
	&\sum_{m=1}^M\sum_{s=1}^{\beta}\sum_{q=Q}^\infty 2p_2e^{(-2q\Delta_M^2)} \leq\sum_{m=1}^M\sum_{s=1}^{\beta}2e^{(-2Q\Delta_M^2)}=2M\beta e^{(-2Q\Delta_M^2)},
	\end{aligned}
	\end{equation}
	where $(m,s)$ refers to channel $m$ with $s$ players on it, (a) results from applying the union bound and (b) from using Hoeffding's inequality \cite{hoeffding}, and $p_2=\textrm{Pr}~(q ~\textrm{observations of $(m,s)$} ~|~ q\geq Q)$.
	
	To ensure $\textrm{Pr} (\bar{A}_k|B_k)$ is lower than $\frac{\gamma}{K}$, $Q$ must satisfy:
	\begin{equation}\label{eq:numberofObs}
	Q\geq \frac{1}{2\Delta_M^2}\log(\frac{2KM\beta}{\gamma}).
	\end{equation}
	Then,
	\begin{equation}
	\textrm{Pr} ~(A|B)= 1-\textrm{Pr}(\bar{A}|B)\geq 1-\sum_{k=1}^K\textrm{Pr}~(\bar{A_k}|B_k)=
	1-\gamma,
	\end{equation}
	leading to all APs having an estimate of every channel satisfying condition (\ref{deltaC}) with probability higher than $1-\gamma$. 
	
	Next, we need to find a time horizon $T_h$ for the exploration phase of the channel allocation step large enough such that all players have $\geq Q$ observations of each arm with probability higher than $1-\gamma$. Note that the length of each  exploration phase $T_{\boldsymbol{\hat{\mu}}}$ does not necessarily satisfy $T_{\boldsymbol{\hat{\mu}}}\geq T_h$. In other words, all players can get  $\geq Q$ observations of each arm with probability higher than $1-\gamma$ after multiple exploration phases.
	
	Let $A_{k,m,s}(t)=1$ if player $k$ observed channel $m$ with $s$ APs on it at timeslot $t$, and 0 otherwise. 
	For $0<\tau<1$, we have:
	\begin{equation}\label{eq:obsWithTime}
	\begin{aligned}
	&\textrm{Pr $\left(\right.$ $k$ has $\leq (1-\tau)T_h \mathbb{E}[A_{k,m,s}]$ observations$\left.\right)$}=\textrm{Pr} \left( \sum_{t=1}^{T_h}A_{k,m,s}(t)\leq (1-\tau)T_h \mathbb{E}[A_{k,m,s}]  \right)=\\
	&\textrm{Pr} \left(e^{\left(-d \sum_{t=1}^{T_h}A_{k,m,s}(t)\right)}\geq e^{\left(-d(1-\tau)T_h \mathbb{E}[A_{k,m,s}]\right)}\right) \overset{\text{(a)}}\leq\frac{\mathbb{E}\left[e^{\left(-d \sum_{t=1}^{T_h}A_{k,m,s}(t)\right)}\right]}{e^{\left(-d(1-\tau)T_h \mathbb{E}[A_{k,m,s}]\right)}},
	\end{aligned}
	\end{equation}
	where $d>0$ and (a) results from applying the Chernoff bound. By noting that all players are randomly and uniformly sampling every channel during the exploration phase, for any $k\in \mathcal{K},m \in \mathcal{S},s \in [\beta]$, $A_{k,m,s}$ are i.i.d. across time. Hence:
	\begin{equation}\label{eq:mgfInd}
	\mathbb{E}\left[e^{\left(-d \sum_{t=1}^{T_h}A_{k,m,s}(t)\right)}\right]=\prod\limits_{t=1}^{T_h}\mathbb{E}\left[e^{\left(-d A_{k,m,s}(t)\right)}\right].
	\end{equation}  
	Moreover, $A_{k,m,s}(t)$ is a Bernoulli random variable that takes the value 1 with probability $p_A$. Therefore, we have:
	\begin{equation}
	\mathbb{E}\left[e^{\left(-d A_{k,m,s}(t)\right)}\right]=
	1+p_A(e^{-d}-1)\overset{\text{(a)}}\leq e^{\left(p_A\left(e^{-d}-1\right)\right)},
	\end{equation}
	where $(a)$ follows since $1+y\leq e^y$. Eq.\ (\ref{eq:mgfInd}) can hence be expressed as:
	\begin{equation}\label{eq:mgfT}
	\begin{aligned}
	\mathbb{E}\left[e^{\left(-d \sum_{t=1}^{T_h}A_{k,m,s}(t)\right)}\right]\leq& e^{\sum\limits_{t=1}^{T_h}\left(p_A\left(e^{-d}-1\right)\right)}=e^{\left(T_h \mathbb{E}[A_{k,m,s}]\left(e^{-d}-1\right)\right)}.
	\end{aligned}
	\end{equation}
	By inserting Eq.\ (\ref{eq:mgfT}) into Eq.\ (\ref{eq:obsWithTime}), we get:
	\begin{equation}\label{eq:boundInt}
	\begin{aligned}
	&\textrm{Pr $\left(\right.$player $k$ has $\leq (1-\tau)T_h \mathbb{E}[A_{k,m,s}]\left.\right)$}\leq e^{\left(T_h \mathbb{E}[A_{k,m,s}]\left(e^{-d}-1\right)\right)+\left(d(1-\tau)T_h \mathbb{E}[A_{k,m,s}]\right)}.
	\end{aligned}
	\end{equation}
	To make the bound as tight as possible, $d$ is chosen such that the right hand side of Eq.\ ({\ref{eq:boundInt}}) is minimized, leading to $d=-\log(1-\tau)$. By substituting $d$ by its value in Eq.\ (\ref{eq:boundInt}), we get:
	\begin{equation}\label{eq:obsWithTimeFinal}
	\begin{aligned}
	&\textrm{Pr $\left(\right.$player $k$ has $\leq (1-\tau)T_h \mathbb{E}[A_{k,m,s}]\left.\right)$}\leq
	e^{\left(-T_h \mathbb{E}[A_{k,m,s}]\left(\tau-(1-\tau)\log(1-\tau)\right)\right)}=\\
	&\left(\frac{e^{-\tau}}{(1-\tau)^{(1-\tau)}}\right)^{\left(T_h \mathbb{E}[A_{k,m,s}]\right)}\overset{\text{(a)}}\leq e^{-\frac{\tau^2}{2}T_h \mathbb{E}[A_{k,m,s}]},
	\end{aligned}
	\end{equation}
	where (a) results from having $(1-\tau)\log(1-\tau)>-\tau+\frac{\tau^2}{2}$, obtained by using a Taylor expansion.
	 
	Taking $\tau=1/2$ and  using a union bound on (\ref{eq:obsWithTimeFinal}), we get:
	\begin{equation}\label{eq:obsWithTimeAll}
	\begin{aligned}
	&\textrm{Pr ($\exists ~ k,m,s$ s.t. $k$ has $\leq \frac{T_h}{2} \mathbb{E}[A_{k,m,s}(t)]$ observations)}\leq KM\beta e^{\left(\frac{-\frac{1}{4}T_h\mathbb{E}[A_{k,m,s}]}{2}\right)},
	\end{aligned}
	\end{equation}
		which is upper bounded by $\gamma$ if $T_h$ satisfies:
	\begin{equation}\label{eq:Th}
	T_h\geq\frac{8}{\mathbb{E}[A_{k,m,s}]}\log\left(\frac{KM\beta}{\gamma}\right).
	\end{equation}
	
	Moreover, the number of observations of each arm during $T_h$ timeslots, $\sum_{t=1}^{T_h}A_{k,m,s}(t)$, must be at least equal to $Q$. Hence we need:
	\begin{equation}
	\begin{aligned}
	\sum_{t=1}^{T_h}A_{k,m,s}(t)>&\frac{T_h}{2}\mathbb{E}[A_{k,m,s}]\geq Q>\frac{1}{2\Delta_M^2}\log\left(\frac{2KM\beta}{\gamma}\right),
		\end{aligned}
	\end{equation}
	which holds if:
	\begin{equation}\label{eq:ThV2}
	\begin{aligned}
	T_h\geq\left\lceil\text{max}\left\{\frac{8}{\mathbb{E}[A_{k,m,s}]}\log\left(\frac{KM\beta}{\gamma}\right),\right.\right.\left.\left.\frac{1}{\Delta_M^2\mathbb{E}[A_{k,m,s}]}\log\left(\frac{2KM\beta}{\gamma}\right)\right\}\right\rceil.
	\end{aligned}
	\end{equation}	
	Note that:
	\begin{equation}\label{eq:expectationBound}
	\begin{aligned} &\mathbb{E}[A_{k,m,s}]=\binom{K-1}{s-1}\left(\frac{1}{M}\right)^s\left(1-\frac{1}{M}\right)^{K-s}\overset{\text{(a)}}\geq\left(\frac{1}{M}\right)^s\left(1-\frac{1}{M}\right)^{K-s}\geq\\
	&\left(\frac{1}{M}\right)\left(\frac{1}{M}\right)^{s-1}\left(1-\frac{1}{M}\right)^{K-1}\left(1-\frac{1}{M}\right)^{1-s}\overset{\text{(b)}}\geq\frac{1}{M e^{\left(\frac{K-1}{M-1}\right)}}\left(M-1\right)^{1-s}\overset{\text{(c)}}\geq\frac{\left(M-1\right)^{1-\beta}}{Me^{\left(\frac{K-1}{M-1}\right)}},
	\end{aligned}
	\end{equation}
	where (a) follows from having  $\binom{K-1}{s-1}\geq 1$, (b) from $(1-\frac{1}{x})^{x-1}\geq\frac{1}{e}$, and (c) from $s\leq \beta$.
	
	Hence, $T_h$ can be re-written as:
	\begin{equation}\label{eq:ThV3}
	\begin{aligned}
	T_h\geq\left\lceil\text{max}\left\{\frac{8 M e^{\left(\frac{K-1}{M-1}\right)}}{\left(M-1\right)^{1-\beta}}\log\left(\frac{KM\beta}{\gamma}\right),\right.\right. \left.\left.\frac{Me^{\left(\frac{K-1}{M-1}\right)}}{\Delta_M^2\left(M-1\right)^{1-\beta}}\log\left(\frac{2KM\beta}{\gamma}\right)\right\}\right\rceil.
	\end{aligned}
	\end{equation}
	Having $T_h$,  the probability of all APs having an estimate of the channel means satisfying  Eq.\ (\ref{deltaC}) is lower bounded by:
	\begin{equation}
	\begin{aligned}
	&\textrm{Pr} (A)= 1-\textrm{Pr} (\bar{A})=1-\left(\textrm{Pr} (\bar{A}|B)\,\textrm{Pr} (B)+\textrm{Pr} (\bar{A}|\bar{B})\,\textrm{Pr} (\bar{B})\right)\\
	&\geq 1-\left(\textrm{Pr}(\bar{A}|B))+\textrm{Pr}(\bar{B})\right)\geq 1-(\gamma+\gamma)=1-\gamma_{e,l}^M.
	\end{aligned}
	\end{equation}
	Since $\Delta_M=\frac{J_M^1-J_M^2}{2KN}\leq \frac{KN-0}{2KN}\leq \frac{1}{2}$, Eq.\  (\ref{eq:ThV3}) is satisfied if:
	\begin{equation}
	T_h=\frac{2Me^{\left(\frac{K-1}{M-1}\right)}}{\Delta_M^2\left(M-1\right)^{1-\beta}}\log\left(\frac{4KM\beta}{\gamma_{e,l}^M}\right).
	\end{equation}
	
	Having found the minimum needed length of the exploration epoch in the channel allocation phase, next, we upper bound the error probability in the $l^{\textrm{th}}$ exploration epoch. To do so, we first note that:
	\begin{equation}
	T_{\boldsymbol{\hat{\mu}_M}}\times l = T_h=\frac{2Me{ \left(\frac{K-1}{M-1}\right)}}{\Delta_M^2\left(M-1\right)^{1-\beta}}\log\left(\frac{4KM\beta}{\gamma_{e,l}^M}\right).
	\end{equation}
	To have $\gamma_{e,l}^M\leq 4KM\beta e^{-l}\leq 4 (M\beta)^2e^{-l}$, the length of each exploration epoch must satisfy: 
	\begin{equation}\label{eq:t_mu}
	T_{\boldsymbol{\hat{\mu}_M}}\geq \frac{2Me^{\left(\frac{K-1}{M-1}\right)}}{\Delta_M^2\left(M-1\right)^{1-\beta}}.
	\end{equation}
	\subsection{Lower Bound of the Length of the Exploration Phase in the Power Allocation Step}\label{appendix2:power}
	By following a similar analysis of the one in Appendix \ \ref{appendix2:channel}, the minimum length of the length of the exploration phase on each channel $m$ in the power allocation step can be given by:
	\begin{equation}\label{eq:Tp}
	T_P^0=
	\left\lceil{\frac{2Le^{\left(\frac{\beta-1}{L-1}\right)}}{\Delta_p^2}}\right\rceil.
	\end{equation}
	If the length of the exploration phase in the power allocation step on each channel $m$ satisfies Eq.\ (\ref{eq:Tp}), then  all players have an estimate of the power level means satisfying the condition in\ (\ref{deltaP}),  with probability $\geq 1-\gamma^P_{e,l}$, where $\gamma^P_{e,l}$ is upper bounded by $ 4\beta L e^{-l}$.
	
	\section{Proof of Lemma\ \ref{lemma3}}\label{appendix3}
	Let $p$ be the true probability of player $k$ not being the sole occupier of some channel $m$ when $k$ accesses the $M$ channels uniformly at random:
	\begin{equation}\label{eq:probNoCollision}
	p=1-\sum\limits_{m=1}^M \frac{1}{M}\left(1-\frac{1}{M}\right)^{K-1}=1-\left(1-\frac{1}{M}\right)^{K-1}.
	\end{equation}
	From Eq. (\ref{eq:probNoCollision}), the number of APs $K$ is given by:
	\begin{equation}
	K=\textrm{round}\left(\frac{\log(1-p)}{\log(1-\frac{1}{M})}+1\right).
	\end{equation}
	The estimated probability of player $k$ not accessing channel $m$ alone at time $t$ is: $
	\hat{p}_t={b_k^t}/{t}.$ 
	For a correct estimation of the number of APs, we need to find a time $t$ sufficiently large to guarantee with high probability that:
	\begin{equation}\label{eq:k=k_hat}
	\begin{aligned}
	\hat{K}=&\textrm{round}\left(\frac{\log(1-\hat{p}_t)}{\log(1-\frac{1}{M})}+1\right)=\textrm{round}\left(\frac{\log(1-p)}{\log(1-\frac{1}{M})}+1\right)=K.
	\end{aligned}	\end{equation}
	To ensure Eq.\ (\ref{eq:k=k_hat}), if $\kappa < 1/2$, the following must hold:
	\begin{equation}\label{eq:p_pt}
	\left|\frac{\log(\frac{t-b_k^t}{t})}{\log(1-\frac{1}{M})}-\frac{\log(1-p)}{\log(1-\frac{1}{M})}\right|=\left|\frac{\log\left(\frac{1-\hat{p}_t}{1-p}\right)}{\log\left(1-\frac{1}{M}\right)}\right|\leq \kappa.
	\end{equation}
	Let $\hat{p}_t-p=\xi$. After some calculations, Eq.\ (\ref{eq:p_pt}) can be expressed as:
	\begin{equation}
	\begin{aligned}
	(1-p)\left(1-\left(1-\frac{1}{M}\right)^{-\kappa}\right)\leq\xi&\leq
	 (1-p)\left(1-\left(1-\frac{1}{M}\right)^\kappa\right)&.
	\end{aligned}
	\end{equation}
	With high probability, $K=\hat{K}$ when $\kappa<\frac{1}{2}$, if $|\hat{p}_t-p|\leq \xi_1$, where:
	\begin{equation}\label{eq:xi1}
	\begin{aligned}
	\xi_1= \min\left\{\left|(1-p)\left(1-\left(1-\frac{1}{M}\right)^{-\kappa}\right)\right|,\right. 
	\left.\left|(1-p)\left(1-\left(1-\frac{1}{M}\right)^\kappa\right)\right|\right\}.
	\end{aligned}
	\end{equation}
	Let $T_{\hat{K}}$ be a large enough time horizon for which the estimated probability $\hat{p}_{T_{\hat{K}}}$ is an average of i.i.d. random variables with expectation $p$. Using Hoeffding's inequality \cite{hoeffding}, we get:
	\begin{equation}
	\textrm{Pr}\left(|\hat{p}_{T_{\hat{K}}}-p|\geq \xi_1\right)\leq 2 e^{-2T_{\hat{K}}\xi_1^2}.
	\end{equation}
	To bound the probability of an incorrect estimation of $\hat{K}$ by some small value $\eta$,  $T_{\hat{K}}$ must be lower bounded by:
	\begin{equation}
	T_{\hat{K}}\geq\frac{\log(2\eta)}{2\xi_1^2}.
	\end{equation} 
	To get a simpler expression of $\xi_1$ and hence of $ T_{\hat{K}}$, suppose that $\kappa=0.49$. With the expression of $p$ given by Eq.\ (\ref{eq:probNoCollision}), the first term in Eq. (\ref{eq:xi1}) can be lower bounded as:
	\begin{equation}\label{eq:lbTerm1}
	\begin{aligned}
	&\left|\left(1-\frac{1}{M}\right)^{K-1}\left(1-\left(1-\frac{1}{M}\right)^{-0.49}\right)\right|=-\left(1-\frac{1}{M}\right)^{K-1}\left(1-\left(1-\frac{1}{M}\right)^{-0.49}\right)\overset{\text{(a)}}\geq\\
	&\left(1-\frac{1}{M}\right)^{M\beta-1}\left(1-\left(-1+\frac{1}{M}\right)^{-0.49}\right)\overset{\text{(b)}}\geq\frac{1}{e^{\left(\frac{M\beta-1}{M-1}\right)}}\left(1-\left(-1+\frac{1}{M}\right)^{-0.49}\right)\overset{\text{(c)}}\geq\frac{0.49}{Me^{\left(\frac{M\beta-1}{M-1}\right)}},
	\end{aligned}
	\end{equation}
	where (a) results from having $M \beta \geq K$, (b) from $(1-\frac{1}{x})^{x-1}\geq \frac{1}{e}$, and (c) from using a Taylor Expansion.
	Similarly, the second term in Eq. (\ref{eq:xi1}) can be lower bounded as:
	\begin{equation}\label{eq:lbTerm2}
	\begin{aligned}
	\left|\left(1-\frac{1}{M}\right)^{K-1}\left(1-\left(1-\frac{1}{M}\right)^{0.49}\right)\right|\geq
	\frac{0.49}{Me^{\left(\frac{M\beta-1}{M-1}\right)}}.
	\end{aligned}
	\end{equation}
	Variable $\xi_1$ is therefore lower bounded by: $
	\xi_1\geq\frac{0.49}{Me^{\left(\frac{M\beta-1}{M-1}\right)}}.$
	Hence, $\hat{K}=K$ with probability higher than $1-\eta$ if:
	\begin{equation}\label{eq:t_k}
	T_{\hat{K}}=
	\left\lceil2.08\log{\left(\frac{2}{\eta}\right)}M^2e^{2\left(\frac{M\beta-1}{M-1}\right)} \right\rceil.
	\end{equation}
\end{appendices}
\vspace{-1cm}
\bibliographystyle{IEEEtran}

\bibliography{IEEEabrv,MABRevised_Arxiv}

\end{document}